\documentclass[a4paper]{llncs}
\spnewtheorem{thm}{Theorem}[section]{\bfseries}{\itshape}

\usepackage{proof}
\usepackage{amssymb}
\usepackage{amsmath}
\usepackage{stmaryrd}
\usepackage[all]{xy}
\DeclareMathAlphabet{\mathsl}{OT1}{cmr}{m}{sl}

\usepackage{color}
\usepackage{url}
\usepackage{ifpdf} 
\usepackage{lineno}
\usepackage{balance}
\usepackage{graphicx}
\usepackage{lipsum}
\usepackage{wrapfig}
\usepackage{microtype}
\usepackage[dvipsnames]{xcolor}
\usepackage{multirow}
\usepackage{booktabs}
\usepackage{times}

\usepackage{framed}
\colorlet{shadecolor}{gray!20}

\usepackage[normalem]{ulem}
\usepackage{enumerate}

\newcommand{\red}[1]{\textcolor{Mahogany}{#1}}
\newcommand{\blue}[1]{\textcolor{MidnightBlue}{#1}}
\long\def\comment#1{}
\newcommand{\eg}{{\em e.g.}}
\newcommand{\ie}{{\em i.e.}}
\newcommand{\etal}{\emph{et al.}}

\newcommand\uSize{\ensuremath{k}}

\newcommand{\Ascr}{\mathcal{A}}
\newcommand{\Pscr}{\mathcal{P}}

\newcommand{\Cscr}{\mathcal{C}}

\newcommand{\Sscr}{\mathcal{S}}
\newcommand{\Rscr}{\mathcal{R}}
\newcommand{\Tscr}{\mathcal{T}}
\newcommand{\Wscr}{\mathcal{W}}

\newcommand{\CS}{\mathcal{CS}}

\newcommand{\lts}{l.t.s.~}
\newcommand{\ltss}{l.t.s.}

\newcommand{\newprop}{reliability}
\newcommand{\Newprop}{Reliability}

\newcommand{\pon}{point-of-no-re\-turn}

\newcommand{\PON}{Point-of-No-Re\-turn}
\newcommand{\pons}{points-of-no-return}
\newcommand{\npon}{breaking point}
\newcommand{\prop}{recoverability}

\newcommand{\PROP}{Recoverability}

\newcommand{\realiz}{$Z$~}
\newcommand{\sur}{$S$~}
\newcommand{\rel}{$L$~}
\newcommand{\recov}{$V$~}
\newcommand{\nrealiz}{\mbox{$n$-$Z$}~}
\newcommand{\nsur}{\mbox{$n$-$S$}~}
\newcommand{\nrel}{\mbox{$n$-$L$}~}

\newcommand{\quotes}[1]{``#1''}

\newcommand\next{\mathcal{N}}
\newcommand\critical{\mathcal{X}}
\newcommand\mustTick{\ensuremath{\mathcal{L}}}

\newcommand{\tup}[1]{\langle#1\rangle}

\newcommand\lra{\longrightarrow}

\newcommand\Dmax{D_{max}}
\usepackage{tweaklist}
\renewcommand{\enumhook}{
\setlength{\topsep}{0em}
\setlength{\itemsep}{0em}
\setlength{\partopsep}{0em}
\setlength{\parskip}{0em}
\setlength{\parsep}{0em}
\addtolength{\leftmargin}{-0.75em}
}
\renewcommand{\itemhook}{
\setlength{\topsep}{0em}
\setlength{\itemsep}{0em}
\setlength{\partopsep}{0em}
\setlength{\parskip}{0em}
\setlength{\parsep}{0em}
\addtolength{\leftmargin}{-0.75em}
}

\title{On the Complexity of Verification of Time-Sensitive Distributed Systems:
Technical Report}

\author{Max Kanovich\inst{1,6} Tajana Ban Kirigin\inst{2} Vivek Nigam\inst{3,4} Andre Scedrov\inst{5} and Carolyn Talcott\inst{7}}
\institute{
University College London, London, UK,
\email{m.kanovich@ucl.ac.uk}
\and
Department of Mathematics University of Rijeka, Rijeka, Croatia,
\email{bank@uniri.hr}
\and
Federal University of Paraíba, Jo\~ao Pessoa, Brazil,
\email{vivek@ci.ufpb.br}
\and
Munich Research Center, Huawei, Munich, Germany
\and
University of Pennsylvania, Philadelphia,
USA,
\email{scedrov@math.upenn.edu}
\and 
HSE University, Computer Science Dept., Moscow, Russia
\and
SRI International, Menlo Park, USA, \email{clt@csl.sri.com} 
}

\begin{document}
\maketitle
\begin{abstract}
This paper develops a Multiset Rewriting language with explicit time for the specification and analysis of Time-Sensitive Distributed Systems (TSDS). 
Goals are often specified using explicit time constraints. A good trace is an infinite trace in which the goals are satisfied perpetually despite possible interference from the environment. In our previous work~\cite{kanovich16formats}, we discussed two desirable properties of TSDSes, \emph{realizability} (there exists a good trace) and \emph{survivability} (where, in addition, all admissible traces are good). Here we consider two additional properties, 
\emph{recoverability} (all compliant traces do not reach points-of-no-return) and \emph{reliability} (the system can always continue functioning using a good trace). Following~\cite{kanovich16formats}, we focus on a class of systems called \emph{ Progressing Timed Systems } (PTS), where intuitively only a finite number of actions can be carried out in a bounded time period. We prove that for this class of systems the properties of recoverability and reliability coincide and are PSPACE-complete.
 Moreover, if we impose a bound on time (as in bounded model-checking), we show that for PTS the reliability property is in the $\Pi_2^p$ class of the polynomial hierarchy, a subclass of PSPACE.
We also show that the bounded survivability is both NP-hard and
coNP-hard.
\end{abstract}

\vspace{1em}

\section{Introduction}

In our previous work~\cite{kanovich16formats}, we considered the verification of Time-Sensitive Distributed Systems (TSDS) motivated by applications with autonomous drones performing surveillance of an area.  The drones must always collectively have recent pictures, \ie, at most M time units old, of certain strategic locations. In attempting to achieve this goal, the drones consume energy and must return to the base station to recharge their batteries. In addition, the environment may interfere as there may be winds that move the drone in a certain direction, or other flying objects may block a drone's path.

In~\cite{kanovich16formats} we considered two verification properties, realizability and survivability. Here we introduce two more properties, reliability and recoverability. Let us explain all four properties in a little more detail.  
The \emph{realizability} problem consists of checking, whether under the given time constraints,
the specified system can achieve the assigned goal, \eg, always collect recent pictures of the sensitive locations. In many settings, the drones themselves or the environment may behave non-deterministically. For example, if a drone wants to reach a point in the northeast, it may initially move either north or east, both being equally likely. Similarly, there could be wind at a particular location, causing any drone under the influence of the wind to move in the direction of the wind.
A stronger property, \emph{survivability}, accounts for such nondeterminism and tests whether the specified system can achieve the assigned goal for all possible outcomes (of drone actions and environmental influences). 
The properties of realizability and survivability represent the two extremes w.r.t. requirements placed on a system. A system that is realizable can achieve the designed goal in some way. A system that satisfies survivability will always achieve the goal, under all circumstances. In some cases, realizability may not be satisfactory, while in others, survivability may be too costly or unattainable. For such systems, intermediate solutions are of interest.

To model such intermediate requirements in system design, in this paper we introduce additional properties, namely \emph{\newprop} ~and \emph{\prop}.  In order to ensure system goals, drones should always be able to function. In particular, drones should always be able to come back to recharge, both in terms of distance and energy. In other words, drones should never go too far and reach so-called \emph{points-of-no-return} where it may no longer be possible to safely return to home base. Engineers should strive to program drones to avoid reaching points-of-no-return. This property is referred to as \emph{recoverability}.

A system satisfies \emph{reliability} if the system is always able to successfully continue its expected performance, \ie, the system never gets stuck. For example, drones should always be able to ensure the system goals, regardless of the disturbances they have experienced in the environment. At any point in time, after the drones have successfully monitored sensitive locations for a certain period of time, they should be able to find a way to continue with their good performance. For example, considering possible technical failures and maintenance of the drones, it may be necessary for engineers to call in additional drones to collectively provide up-to-date images of the entire area of interest.

Following \cite{kanovich16formats}, we focus on a class of systems called \emph{ Progressing Timed Systems } (PTS), which are specified as timed multiset rewriting theories. In a PTS, only a finite number of actions can be carried out in a bounded time interval. 
In addition to formalizing the properties, we show that the following relations hold for PTS:
$$
S_{urvivability}
\ \implies \ R_{eliability}
\ \Longleftrightarrow  \ R_{ecoverability} 
\ \implies  \ R_{ealizability}\ . 
$$

In their spirit, these properties seem similar to safety and liveness properties~\cite{alpern87dc} or a combination of these properties. However, it is not straightforward to classify them in these terms. 
The properties we consider, defined in Section~\ref{sec:problems}, contain an alternation of quantifiers, which makes it more challenging to formally represent them as a combination of safety and liveness properties~\cite{alpern87dc}. 

In our previous work~\cite{kanovich.mscs,kanovich12rta,kanovich15post,kanovich17jcs}, we proposed a timed Multiset Rewriting (MSR) framework for specifying compliance properties similar to {quantitative safety properties}~\cite{alpern87dc,clarkson10jcs} and investigated the complexity of a number of decision problems. These properties were defined over sets of {finite traces}, \ie, executions of a finite number of actions. 
The above properties, on the other hand, are defined over \emph{infinite traces}. 

The transition to properties over infinite traces leads to many challenges, as one can easily fall into undecidable fragments of verification problems. The main challenge is to identify the syntactic conditions on specifications so that the verification problems fall into a decidable fragment and, at the same time, that interesting examples can be specified.

The  remainder of the paper is organized as follows: 

\begin{itemize}
\item
Following \cite{kanovich16formats}, in Section \ref{sec:timedmsr} we discuss  \emph{Progressing Timed Systems} (PTS). 
 In Section~\ref{sec:progdrones} we illustrate its expressiveness by encoding a simplified drone example as a PTS.

\item 
In Section~\ref{sec:timedprop} we define a language for specifying the relevant quantitative temporal properties of timed systems used to define the properties of \emph{realizability}, \emph{ \newprop, \prop} and \emph{survivability}.
\item
In Section~\ref{sec:relations-properties} we then formally compare the expressiveness of these properties.

\item Section~\ref{sec:complex} investigates
the complexity of verification problems that involve the above properties.
 While these problems are undecidable in general~\cite{kanovich.mscs}, we show that they are PSPACE-complete for PTSes. 
We also show that, when we bound time (as in bounded-model checking),  realizability of PTSes is NP-complete, survivability is in the 
$\Delta_2^p$ class of the polynomial hierarchy and the \newprop~is in the $\Pi_2^p$ class of the polynomial hierarchy~\cite{papadimitriou07book}. 
The upper bound results regarding realizability and survivability were obtained in \cite{kanovich16formats}, while here we obtain new complexity results for  the lower bound complexity results for the $n$-time bounded survivability and the  complexity results relating to reliability from Section~\ref{sec:complex}.
\item
We also provide a discussion on related and future work, Section~\ref{sec:related}. 

\end{itemize}

\paragraph{Relation to our previous work}
This technical report considerably extends the conference paper~\cite{kanovich16formats}. It also updates and subsumes the technical report \cite{kanovich16arxivformats}.
For ease of reference, we include some of the material from~\cite{kanovich16formats,kanovich16arxivformats}.
All the material involving properties of \newprop~and \prop~is new, including the investigation of the relations among all four properties from Section~\ref{sec:relations-properties}, the complexity results relating to reliability from Section~\ref{sec:complex}, and the lower bound complexity results for $n$-time bounded survivability are new.

\section{Multiset Rewriting Systems}
\label{sec:timedmsr}

Assume a finite first-order typed alphabet, $\Sigma$, with variables, constants, function and predicate symbols.
Terms and formulas are constructed as usual (see~\cite{enderton}) by applying symbols of correct type (or sort). 

\begin{definition}[Fact]
If $P$ is a predicate of type $\tau_1 \times \tau_2 \times \cdots \times \tau_n \rightarrow o$, where $o$ is the type for propositions, and $u_1, \ldots, u_n$ are terms of types $\tau_1, \ldots, \tau_n$, respectively, then $P(u_1, \ldots, u_n)$ is a \emph{fact}.
A fact is \emph{ground} if it contains no variables. 
\end{definition}

We assume that the alphabet contains the constant $z : Nat$ denoting zero and the function $s : Nat \to Nat$ denoting the successor function. Whenever it is clear from the context, we write $n$ for $s^n(z)$ and $(n + m)$ for $s^n(s^m(z))$. 

\vspace{0.5em}
In addition, we allow an unbounded number of fresh values~\cite{cervesato99csfw,durgin04jcs} to be involved.

\vspace{0.5em}
In order to specify timed systems,  we attach a timestamp 
to each fact.

\begin{definition}[Timestamped Fact]
\emph{Timestamped facts} are of the form $F@t$, where $F$ is a fact and $t \in \mathbb{N}$ is a natural number called {\em timestamp}.
\end{definition}

Note that timestamps are \emph{not} constructed by using the successor function. To obtain the complexity results, we use a symbolic representation of the problems and abstractions that can handle unbounded time values. For more insight see discussion after Definition~\ref{def:size-fact}.

There is a special predicate symbol $Time$ with arity zero that is used to represent global time.

For simplicity, we often just say facts instead of timestamped facts. Also, when we want to emphasize the difference between a fact $F$ and a timestamped fact $F@t$, we say that $F$ is an \emph{untimed fact}.

\begin{definition}[Configuration]
A {\em configuration} is a finite multiset of ground timestamped facts, 
~$\Sscr = \{~Time@t, ~F_1@t_1, \ldots, ~F_n@t_n~\}$ ~
with a single occurrence of a $Time$ fact. \\[3pt]
Given a configuration $\Sscr$ containing $Time@t$, we say that a fact $F@t_F$ in $\Sscr$ is a \emph{future fact} if its timestamp is greater than the global time $t$, \ie,~if ~$t_F>t$.
Similarly, a fact $F@t_F$ in $\Sscr$ is a \emph{past fact} if ~$t_F<t$, and a fact $F@t_F$ in $\Sscr$  is a \emph{present fact} if ~$t_F=t$.
\end{definition}

Configurations are to be interpreted as states of the system. Consider the following configuration where the global time is 4:
\[
\begin{small}
\Sscr_1 = \left\{\begin{array}{l}
Time@4, \,Dr(d1,1,2,10)@4,\,Dr(d2,5,5,8)@4,\,P(p1,1,1)@3,\,P(p2,5,6)@0
\end{array}\right \}
\end{small}
\label{conf-example-1}
\]
Fact $Dr(d_{Id},x,y,e)@t$ denotes that drone $d_{Id}$ is at position $(x,y)$ at time $t$ with $e$ energy units left in its battery; fact $P(p_{ID},x,y)@t$ denotes that a point to be monitored  is at position $(x,y)$ and that the last picture of it was taken at time $t$. Thus, the above configuration denotes a scenario with two drones located at positions $(1,2)$ and $(5,5)$ and with 10 and 8 energy units, and with two points to be monitored at positions $(1,1)$ and $(5,6)$, where the former was last photographed at time $3$ and the latter at time 0. 

Using variables, including time variables, we are able to represent  (sets of) configurations of particular form. 
 For example, 
$$~Time@(T+D), \,Dr(X,5,6,Y)@(T+D),\,P(p2,5,6)@T$$
specifies that some drone $X$ with $Y$ energy units is currently  at the position $(5,6)$ and that the point of interest at position $(4,6)$ was last photographed $D$ time units ago. This holds for any configuration containing the above facts for some instantiation of the variables $T,D,X$ and $Y$.

\vspace{0.5em}
Configurations are modified by multiset rewrite rules which can be interpreted as actions of the system. 
There is only one rule, $Tick$, which represents how global time advances
\begin{equation}
\label{eq:tick}
Time@T \lra Time @ (T+1)
\end{equation}
where $T$ is a time variable denoting the global time. 
With an application of a $Tick$ rule,  a configuration
$\{~Time@t, \,F_1@t_1, \ldots, \,F_n@t_n~\}$ representing the  state of a system at time \,$t$, is replaced with the configuration 
 \mbox{$\{~Time@(t +1 ), \,F_1@t_1, \ldots, \,F_n@t_n~\} $} \ representing the system at time ~$t+1$.

\vspace{0.5em}
The remaining rules are \emph{instantaneous}, since they do not modify global time, but may modify the remaining facts of configurations (those different from $Time$). Instantaneous rules have the form:
\begin{equation}
\begin{array}{l}
Time@T, \, 
\,W_1@T_1,\ldots,\,W_p@T_p,
\red{\,F_1@T_1'}, \ldots, \red{\,F_n@T_n'} \ \mid \ \,\Cscr \ \lra \\ \ 
\exists \vec{X}.\,[ \ Time@T,
\,W_1@T_1,\ldots,\,W_p@T_p, \,\blue{Q_1@(T + d_1)}, \ldots, \,\blue{Q_m@(T + d_m)} \, ]
\label{eq:instantaneous}
\end{array}
\end{equation}
where $d_1, \ldots, d_m$ are natural numbers, 
\, $W_1@T_1,\ldots,\,W_p@T_p, \, F_1@T_1', \ldots, {\,F_n@T_n'}$
are timestamped facts, possibly containing variables, and $\Cscr$ is the guard of the rule which is a set of constraints
involving the time variables that appear as timestamps of facts in the pre-condition of the rule, \ie,~the variables ~$T, T_1, \ldots, T_p, T_1', \ldots, T_n'$. 
The facts $W_i, F_j$ and $Q_k$ are all different from the fact $Time$ and $\vec{X}$ are variables that do not appear in $W_1,\ldots,\,W_p,\, F_1, \ldots, {\,F_n}$.

Constraints may be of the form:
\begin{equation}
\label{eq:constraints}
T > T' \pm d \quad \textrm{ or } \quad T = T' \pm d 
\end{equation}
where $T$ and $T'$ are time variables, and $d\in\mathbb{N}$ is a natural number. 

{Here and throughout the rest of the paper, the symbol $\pm$ stands for either $+$ or $-$, \ie, constraints may involve addition or subtraction.}

We use $T' \geq T' \pm d$ to denote the disjunction of $T > T' \pm d$ and $T = T' \pm d$. 
All variables in the guard of a rule are assumed to appear in the rule's pre-condition. 

\vspace{0,5em}
Finally, the variables $\vec{X}$ that are existentially quantified in a rule (Eq.~\ref{eq:instantaneous})
are to be replaced by fresh values, also called \emph{nonces} in the protocol security literature~\cite{cervesato99csfw,durgin04jcs}. 
As in our previous work~\cite{kanovich13ic}, we use nonces whenever unique identification is required, for example for drone identification.

Let \,$\Wscr$ \, and \, $\Wscr'$ \, be multisets of timestamped facts.
A rule 
~$\Wscr \mid \Cscr \lra \exists \vec{X}.\, \Wscr'$~
can be applied to a configuration $\Sscr$ if there is a ground substitution $\sigma$ such that \, $\Wscr\sigma \subseteq \Sscr$ \, and that\, $\Cscr\sigma$~is true.
The resulting configuration is~ 
$$\big(\,(\Sscr \setminus \Wscr) \cup \Wscr'\, \big)\sigma \ ,$$
where variables  $\vec{X}$ are fresh.
More precisely, given a rule $r$, an instance of a rule is obtained by substituting constants for all variables appearing in the pre- and post-condition of the rule. This substitution applies to variables appearing in terms inside facts, to variables representing fresh values, and to time variables used to specify timestamps of facts.

An instance of an instantaneous rule can only be applied if all the constraints in its guard are satisfied. 
For example, since~$ 0+2<5$ (when instantiating $T'$ as the timestamp of the fact $\,P(p2,5,6)@0$\,)  rule 
\[
\begin{small}
\begin{array}{l}
Time@T, \red{P(I,X,Y)@T'}, \red{Dr(Id,X,Y,E+1)@T} \mid \{~ T'+2<T~\}%
\lra
\\ \qquad \qquad \qquad Time@T, \blue{P(I,X,Y)@T}, \blue{Dr(Id,X,Y,E)@(T+1)}
\end{array}
\end{small}
\]
is applicable to configuration 
$$
\{\ Time@5, \,Dr(d1,1,2,10)@5,\,Dr(d2,5,6,7)@5,\,P(p1,1,1)@3,\,P(p2,5,6)@0~\} \, , 
$$
resulting in configuration \ 
$$
\{\ Time@5, \,Dr(d1,1,2,10)@5,\,Dr(d2,5,6,6)@6,\,P(p1,1,1)@3,\,P(p2,5,6)@5~\} \, , 
$$
but it is not applicable to the following configuration  
$$
\{\ Time@5, \,Dr(d1,1,2,10)@5,\,Dr(d2,5,5,8)@5,\,P(p1,1,1)@3,\,P(p2,5,6)@4 \ \} \ 
$$
because there are no $\Pscr(p,x,y)@T'$ facts in the configuration such that its timestamp $T'$ satisfies the given constraint, $~ T'+2<T$, involving the global time $T$. Namely, ~$3+2~{\not}{<}~5$ and ~$4+2~{\not}{<}~5$.

\vspace{0,5em}
Following \cite{durgin04jcs} we say that a timestamped fact $F@T$  is \emph{consumed} by a rule $r$ if that fact occurs more times on the left side than on the right side of the rule $r$. A timestamped fact $F@T$  is \emph{created} by some rule $r$ if this fact occurs more times on the right side than on the left side of the rule $r$. 
Hence, facts ${F_1@T_1'}, \ldots, {F_n@T_n'}$ are consumed by  rule (Eq.~\ref{eq:instantaneous}) while facts ${Q_1@(T + d_1)}, \ldots, {Q_m@(T + d_m)}$ are created by this rule.
Note that a fact $F$ can appear in a rule with different timestamps, but for the above notions we count instances of the same timestamped fact $F@T$.
 In a rule, we usually color \red{red} the consumed facts and \blue{blue} the created facts.

\begin{remark}
Using constraints we are able to formalize time-sensitive 
properties and problems that involve explicit time requirements. 
The set of constraints may, however,  be empty, \ie, rules may have no constraints attached.
\end{remark}

\vspace{1mm}
We write ~$\Sscr \lra_r \Sscr'$\ for the one-step relation where the configuration $\Sscr$ is rewritten into $\Sscr'$ using an instance of rule $r$. 
For a set of rules $\Rscr$, we define ~$\Sscr \lra_\Rscr^* \Sscr'$~ to be  the transitive reflexive closure of the one-step relation on all rules in $\Rscr$. We omit the subscript \ $\Rscr$, when it is clear from the context, and simply write ~$\Sscr \lra^* \Sscr'$.

\vspace{2mm}
Note that due to the nature of multiset rewriting, there are various aspects of non-determinism in the model. For example, different actions and even different instantiations of the same rule may apply to the same configuration $\Sscr$, leading to different resulting configurations $\Sscr'$. 

\begin{definition}[Timed MSR System]
A \emph{timed MSR system} 
$\Tscr$ is a set of rules containing only instantaneous rules (Eq.~\ref{eq:instantaneous}) and the $Tick$ rule (Eq.~\ref{eq:tick}).
\end{definition}

A trace of a timed MSR system 
is constructed by a sequence of its rules. 
In this paper, we consider both finite and infinite traces. 
A \emph{finite trace} of a timed MSR system $\Tscr$ starting from an initial configuration $\Sscr_0$ is a sequence
$$
\Sscr_0 \lra \Sscr_1 \lra \Sscr_2 \lra \cdots \lra \Sscr_n 
$$
and an \emph{infinite trace} of $\Tscr$ starting from an initial configuration $\Sscr_0$ is a sequence
$$
\Sscr_0 \lra \Sscr_1 \lra \Sscr_2 \lra \cdots \lra \Sscr_n \lra \cdots 
$$
where for all ~$i \geq 0$, \ $\Sscr_{i} \lra_{r_i} \Sscr_{i+1}$ \ for some $r_i \in \Tscr$.  When a configuration $\Sscr$ apperas in a trace $P$ we write ~$\Sscr \in P$.

\vspace{0,5em}
We will pay particular attention to periods of time represented by traces. 
Since time advances by one unit of time per $Tick$ rule, a finite (infinite) number of $Tick$ rules in a trace represents a finite (infinite) time period.
One can easily imagine traces containing a finite number of $Tick$ rules and an infinite number of instantaneous rules. Such traces would represent an infinite number of actions performed in a finite time interval. 
In this paper we are not interested in such traces and focus on so called \emph{infinite time traces}.

\begin{definition}[Infinite Time Trace]
\label{def:infinity}
A trace $P$ of a timed MSR 
$\Tscr$ is an ~\emph{infinite time trace} ~if the time tends to infinity in $P$, \ie, 
$(\forall n \in \mathbb{N}) \ (\exists~\Sscr\in P)$ such that \ $Time@T \in \Sscr $ and ~$T> n$. 
\end{definition}

Since in any trace, the global time ticks in single time units, it follows immediately that any infinite time trace is an infinite trace, and it contains an infinite number of $Tick$ rules.

We have shown in our previous work~\cite{kanovich11jar,kanovich13ic,kanovich.mscs,kanovich15post,kanovich17jcs} that problems involving MSR, such as checking whether a configuration can be reached, are undecidable if no further restrictions are imposed. These problems are undecidable already when considering only finite traces. However, these problems are decidable for balanced MSR systems~\cite{kanovich11jar,kanovich.mscs} that assume an upper-bound, $k$, on the size of facts formally defined below.

\begin{definition}[Balanced System]
\label{def:balanced}
A timed MSR system $\Tscr$ is \emph{balanced} if for all instantaneous rules $r \in \Tscr$, $r$ creates the same number of facts as it consumes, \ie, the instantaneous rules are of the form: 
\begin{equation}
\begin{array}{l}
Time@T, \,\Wscr, \red{\,F_1@T_1'}, \ldots, \red{\,F_n@T_n'} \ \mid \ \,\Cscr \ \lra \\ \ \
\exists \vec{X}.~[ \ Time@T,\, \Wscr, \blue{\,Q_1@(T + d_1)}, \ldots, \blue{\,Q_n@(T + d_n)} \, ] \ .
\label{eq:instantaneous-balanced}
\end{array}
\end{equation}
\end{definition}

By consuming and creating facts, rewrite rules can increase and decrease the number of facts in configurations throughout a trace. 
However, in balanced MSR systems, 
rule application does not affect the number of facts in a configuration. That is, enabling configuration has the same number of facts as the resulting configuration.
Hence, the number of facts in configurations is constant throughout a trace.

\begin{definition}[Size of a Fact]\label{def:size-fact}
The size of a timestamped fact $P@T$, written $|P@T|$ is the total number of alphabet symbols appearing in $P$.
\end{definition}

 For instance, $|P(s(z),f(a,X), a)@12| = 7$. For our complexity results, we assume a bound, $k$, on the size of facts. 
 Without this bound (among other restrictions), any interesting decision problem is shown undecidable by encoding the Post correspondence problem~\cite{durgin04jcs}. 
Note that the value of the timestamp is not considered in the size of facts. For the complexity results, the (unbounded) time values of timestamps are handled using the abstractions and the symbolic representation of the problems.

\subsection{Progressing Timed Systems} 
\label{sec:pts}

Following \cite{kanovich16formats}, we discuss a particular class of timed MSR systems, called \emph{progressing timed MSR systems}~(PTSes), in which only a finite number of actions can be carried out in a bounded time interval. This is a natural condition for many systems, similar  similar to the \emph{finite-variability assumption} used in the temporal logic and timed automata literature.

\begin{definition}[Progressing Timed System]
\label{def:progressing}
A timed MSR system $\Ascr$ is a \emph{progressing timed MSR system (PTS)} if $\Tscr $ is balanced and for all instantaneous rules $r \in \Tscr$:
\begin{itemize}
\item[i)] Rule $r$ creates \emph{at least one} fact with timestamp greater than the global time, \ie, in (Eq.~\ref{eq:instantaneous}), ~ $d_i \geq 1$~ for at least one ~$i \in \{1, \dots, n \}$;

\item[ii)] 
Rule $r$ consumes \emph{only} facts with timestamps in the past or at the current time, \ie, in (Eq.~\ref{eq:instantaneous}), the set of constraints ~$\Cscr$ contains the set 
$$\Cscr_r = \{~T \geq T_i' \mid F_i@T_i', ~1 \leq i \leq n~\} \ . $$
\end{itemize} 
\end{definition}

For the sake of readability, from this point on we assume that for all rules $r$ the set of their constraints implicitly contains the set ~$\Cscr_r$, as shown in Definition~\ref{def:progressing}, and do not always write ~$\Cscr_r$ explicitly in our specifications.

\vspace{0,5em}
The following rule, which denotes the action of a drone taking a photo of a point of interest, is an example of a rule in a PTS:
\[
\begin{array}{l}
Time@T, \red{~P(I,X,Y)@T'}, \red{~Dr(Id,X,Y,E+1)@T} ~ \mid ~ \{~T'<T ~\}
\lra
\\ \qquad \qquad \qquad Time@T, \blue{~P(I,X,Y)@T}, \blue{~Dr(Id,X,Y,E)@(T+1)}
\end{array}
\]
Note that the constraint $T' <T$ is used to prevent drones from repeatedly photographing the same point of interest at the same time to save energy. Also, the created future fact prevents the same drone from performing the same action in the same time unit.

\vspace{0,5em}
The following proposition~\cite{kanovich16formats} establishes a bound on the number of instances of instantaneous rules appearing between two consecutive instances of $Tick$ rules in a trace of a PTS. This bound is then used to formalize the intuition that PTSes always move things forward.

\begin{proposition}
\label{prop:bounded-length}
Let $\Tscr$ be a PTS, 
$\Sscr_0$ an initial configuration and $m$ the number of facts in $\Sscr_0$. For all traces $\Pscr$ of $\Tscr$ starting from $ \Sscr_0$, let 
$$
\Sscr_i \lra_{Tick}  \Sscr_{i+1} \lra \cdots \lra \Sscr_j \lra_{Tick}  \Sscr_{j+1} 
$$
be any subtrace of ~$\Pscr$ with exactly two instances of the $Tick$ rule, one at the beginning and the other at the end. Then ~$j - i < m$.
~\cite{kanovich16formats}
\end{proposition}
\begin{proof}
The statement easily follows from Definition~\ref{def:progressing}. 
Let $\Pscr$ be an arbitrary trace in $\Tscr$ and 
$$
\Sscr_i \lra_{Tick}  \Sscr_{i+1} \lra \cdots \lra \Sscr_j \lra_{Tick}  \Sscr_{j+1} 
$$
an arbitrary subtrace of $\Pscr $ with exactly two instances of the $Tick$ rule.
All the rules between $Tick$ rules in the above subtrace are instantaneous. 

Since $\Tscr $ is a PTS, the application of any instantaneous rule creates at least one future fact 
and consumes at least one present or past fact. 
In other words, an application of an instantaneous rule reduces the total number of past and present facts in the configuration.

Since the system $\Tscr$ is balanced, all the above configurations $\Sscr_i, \dots, \Sscr_j$ have the same number of facts, $m$.
Recall also that the fact $Time$ does not change when the instantaneous rules are applied. 
Thus, since there are at most $m-1$ present or past facts different from $Time$ in any $\Sscr_k$, $i<k\leq j$,  a series of at most $m-1$ instantaneous rules can be applied between two $Tick$ rules.
\qed
\end{proof}

According to the above statement, in a PTS an unbounded number of instantaneous rules cannot be applied in a bounded interval of time.
Also, from the above result we can conclude that infinite traces in PTSes represent infinite time periods. In particular, this means that in traces of PTSes there are no phenomena similar to Zeno paradox.
This is stated in the following proposition.

\vspace{0.5em}

\begin{proposition} 
\label{prop:progressing}
Let $\Tscr $ be a PTS.  
All infinite traces of $\Tscr$ are infinite time traces, \ie, traces where time tends to infinity.
~\cite{kanovich16formats}
\end{proposition} 
\begin{proof}
Assume that in some infinite trace $\Pscr$ of a PTS $\Tscr$ the current time does not exceed some value $M$. 
Then, since timestamps are natural numbers, and time advances by a single time unit,
there are at most $M$ time ticks in $\Pscr$. 

According to Proposition \ref{prop:bounded-length} there are at most $m-1$ instantaneous rules between any $Tick$ rule and the next $Tick$ rule in $\Pscr$. 

Consequently, in total, there are at most $(M+1)\cdot (m-1)+M$ rules in $\Pscr$, \ie,~$\Pscr$ is a finite trace. Contradiction.
\qed
\end{proof}

\vspace{1em}
Finally, notice that the PTS model has many syntactic conditions, \eg, balanced condition (Definition~\ref{def:balanced}), the form of time constraints (Eq.~\ref{eq:constraints}), the form of instantaneous rules (Eq.~\ref{eq:instantaneous}). Each of these conditions has been carefully developed.
As we have shown in our previous work~\cite{kanovich.mscs}, relaxing any of these conditions leads to undecidability of
important verification problems, such as the reachability problem, over finite traces. 
Clearly, these conditions are also needed for infinite traces.

The additional challenge in allowing infinite traces is to make sure that time advances in such a way that traces represent arbitrarily large time periods. Our definition of PTS is a simple and elegant way to enforce this. Moreover, as we show in Section~\ref{sec:progdrones}, it is still possible to specify many interesting examples with our PTS model, including our motivating example, and still prove the decidability of our verification problems involving infinite traces (Section~\ref{sec:complex}).

\vspace{1em}
\section{Programming Drone Behavior using PTS}
\label{sec:progdrones}

\begin{figure}[t]
\begin{scriptsize}
\[
\begin{array}{l}
Time@T, \,\Pscr(p_1,\ldots,p_n), \red{\,Dr(Id,X,Y,E+1)@T} \, \mid \, doMove\,(Id,X,Y,E+1,T,T_1,\ldots,T_n,north)  \\
\qquad \qquad \lra Time@T, \,\Pscr(p_1,\ldots,p_n), \blue{\,Dr(Id,X,Y+1,E)@(T+1)}\\[5pt]

Time@T, \, \Pscr(p_1,\ldots,p_n), \red{\,Dr(Id,X,Y+1,E+1)@T} 
 \lra \\
\qquad \qquad
\mid doMove\,(Id,X,Y+1,E+1,T,T_1,\ldots,T_n,south) \lra \\
\qquad \qquad \qquad Time@T, \,\Pscr(p_1,\ldots,p_n), \blue{\,Dr(Id,X,Y,E)@(T+1)}\\[5pt]

Time@T, \,\Pscr(p_1,\ldots,p_n), \red{\,Dr(Id,X+1,Y,E+1)@T}
 \lra \\
\qquad \qquad\mid doMove\,(Id,X+1,Y,E+1,T,T_1,\ldots,T_n,west) \lra \\
\qquad \qquad \qquad  Time@T, \,\Pscr(p_1,\ldots,p_n), \blue{\,Dr(Id,X,Y,E)@(T+1)}\\[5pt]

Time@T, \,\Pscr(p_1,\ldots,p_n), \red{\,Dr(Id,X,Y,E+1)@T} \mid doMove\,(Id,X,Y,E+1,T,T_1,\ldots,T_n,east) \\
\qquad \qquad \lra Time@T, \,\Pscr(p_1,\ldots,p_n), \blue{\,Dr(Id,X,Y,E)@(T+1)}\\[5pt]

Time@T, \,\Pscr(p_1,\ldots,p_n), \red{\,Dr(Id,x_b,y_b,E)@T} \mid doCharge\,(Id,E,T,T_1,\ldots,T_n) \lra \\
\qquad \qquad Time@T, \,\Pscr(p_1,\ldots,p_n), \blue{\,Dr(Id,x_b,y_b,E+1)@(T+1)}\\[5pt]

Time@T, \,P(p_1,X_1,Y_1)@T_1, \ldots, \red{\,P(p_i,X,Y)@T_i}, \ldots, P(p_n,X_n,Y_n)@T_n, \red{\,Dr(Id,X,Y,E)@T} \\[2pt]
\qquad \mid \ doClick\,(Id,X,Y,E,T,T_1,\ldots,T_i,\ldots,T_n) \lra  Time @T, \\[3pt]
\qquad \ \ \ P(p_1,X_1,Y_1)@T_1, \ldots, \blue{\,P(p_i,X,Y)@T}, \ldots, \,P(p_n,X_n,Y_n)@T_n, \blue{\,Dr(Id,X,Y,E-1)@(T+1)}\\[5pt]

Time@T, \red{\,Dr(Id,X,Y,E)@T} \mid  hasWind\,(X,Y,north) \lra Time@T, \blue{Dr(Id,X,Y+1,E)@(T+1)}\\[5pt]

Time@T, \red{\,Dr(Id,X,Y+1,E)@T} \mid  hasWind\,(X,Y,south) \lra Time@T, \blue{Dr(Id,X,Y,E)@(T+1)}\\[5pt]

Time@T, \red{\,Dr(Id,X+1,Y,E)@T} \mid  hasWind\,(X,Y,west) \lra Time@T, \blue{\,Dr(Id,X,Y,E)@(T+1)}\\[5pt]

Time@T, \red{\,Dr(Id,X,Y,E)@T} \mid  hasWind\,(X,Y,east) \lra Time@T, \blue{\,Dr(Id,X+1,Y,E)@(T+1)}
\end{array} 
\] 
\end{scriptsize}
\vspace{-2mm}
\caption{Macro rules specifying the scenario where drones take pictures of points of interest. Here $\Pscr(p_1,\ldots,p_n)$~ denotes ~$P(p_1,X_1,Y_1)@T_1, \ldots, P(p_n,X_n,Y_n)@T_n$. Moreover, we assume that the Drone stay in a grid of size ~$x_{max} \times y_{max}$~ and have at most~ $e_{max}$ ~energy units.}
\label{fig:rules-complete}
\end{figure}

Following \cite{kanovich16formats}, Figure~\ref{fig:rules-complete} depicts the macro rules of our motivating scenario where drones are moving on a fixed grid of size $x_{max} \times y_{max}$, have at most $e_{max}$ energy units and take pictures of some points of interest. We assume that there are $n$ such points $p_1, \ldots, p_n$, where $n$ is fixed, a base station is at position $(x_b,y_b)$, and that the drones should regularly take pictures so that all pictures are recent. That is, at any time, each of the points of interest should have been photographed in the last $M$ time units, for some given $M$.

\vspace{0.5em}
Clearly if drones non-deterministically choose to move in some direction without a particular strategy, they will fail to achieve the assigned goal. A strategy of a drone can be specified using time constraints. 

For this example, the strategy would depend on the difference $T-T_i$, for $1 \leq i \leq n$, specifying the elapsed time since the last picture of the point $p_i$ was taken. This can be specified with the following
set of time constraints:
\[
\Tscr(d_1,\ldots,d_n) = \{~T - T_1 = d_1, \ldots, T - T_n = d_n~\}
\]
where for all $1 \leq i \leq n$ we instantiate $d_i$ by values in $\{0,\ldots,M\}$.

For example, the macro rule with ~$doMove\,(Id,X,Y,E+1,T,T_1,\ldots,T_n,north)$ in Figure~\ref{fig:rules-complete} is replaced by the set of rules:
\[
\begin{small}
\begin{array}{l}
Time@T, \, \Pscr(p_1,\ldots,p_n)@T, \red{\,Dr(d1,0,0,1)@T} \ 
\\ \qquad \qquad \qquad \qquad \qquad \qquad\qquad \ \
\mid \ \Tscr(0,\ldots,0), \,DoMv\,(d1,0,0,1,0,\ldots,0,north) 
\\[2pt] 
\quad \qquad 
\lra Time@T, \, \Pscr(p_1,\ldots,p_n)@T, \blue{\,Dr(Id,0,1,0)@(T+1)}\\[5pt]

Time@T, \,\Pscr(p_1,\ldots,p_n)@T, \red{\,Dr(d1,0,0,1)@T} \ 
\\ \qquad \qquad \qquad \qquad\qquad \qquad\qquad \ \
\ \mid \ \Tscr(0,\ldots,1), \,DoMv\,(d1,0,0,1,0,\ldots,1,north) 
\\[2pt] 
\quad \qquad 
\lra Time@T, \,\Pscr(p_1,\ldots,p_n)@T, \blue{\,Dr(Id,0,1,0)@(T+1)}\\
\qquad
\cdots\\
Time@T, \,\Pscr(p_1,\ldots,p_n)@T, \red{\,Dr(d2,x_{max},y_{max}-1,e_{max})@T}\\
\qquad \qquad \qquad \mid \ \Tscr(M,\ldots,M), \,DoMv\,(d2,x_{max},y_{max}-1,e_{max},M,\ldots,M,north) \\[2pt]
\quad \qquad 
\lra \ Time@T, \,\Pscr(p_1,\ldots,p_n)@T, \blue{\,Dr(Id,x_{max},y_{max},e_{max}-1)@(T+1)}\\
\end{array}
\end{small}
\]
where $doMove$ returns a tautology or an unsatisfiable constraint depending on the desired behavior of the drone.

Finally, macro rules for moving the drone, taking a picture, charging, and macro specifying winds are similarly defined. 

While most of the rules have the expected result, we only explain the click and wind rules. The click rule is applicable if the drone is at the position of some point of interest. 
If applied, the timestamp of the fact $P(p_i,X,Y)$ is updated to the current time $T$. The wind rule is similar to the move rules moving the drone to some direction, but does not cause the drone to consume its energy.

In our implementation in~\cite{kanovich16formats} we used a more sophisticated approach described in~\cite{talcott16quanticol} using soft-constraints to specify a drone's strategy. It can be translated into a PTS that incorporates the strategy described above.


\paragraph{Other Examples} \ 
Besides examples involving drones, other exampels also seem to be progressing. For example, in our previous work~\cite{kanovich.mscs}, we specify a monitor for clinical trials using our timed MSR system framework with discrete time. This specification is progressing.

There are a number of other examples which we have been investigating and that are progressing. For example, ~\cite{talcott15wirsing} models a simplified version of a package delivery systems inspired by Amazon's Prime Air service, and ~\cite{talcott16quanticol} models a patrolling bot which moves from one point to another. All these examples seem to be progressing.

\section{Quantitative Temporal Properties}
\label{sec:timedprop}
Following \cite{kanovich16formats}, we begin the Section~\ref{subsec:critical} by discussing critical configurations,  a language used to define desirable properties of systems. This is a key concept in our framework, used to describe explicit timing constraints that a system should satisfy.
In Section~\ref{subsec:time-sampling} we discuss lazy time sampling, which is a condition on traces that intuitively enforces that systems react at the expected time. Then in Section~\ref{sec:problems}, we discuss a number of verification problems.

\subsection{Critical Configurations and Compliant Traces}
\label{subsec:critical}
Critical configurations specifications are used for specifying bad configurations that should be avoided by a system. 

\begin{definition}[Critical Configuration]
\emph{Critical configuration specification} is a set of pairs 
$$\CS = \{~\tup{\Sscr_1, \Cscr_1}, \ldots, \tup{\Sscr_n, \Cscr_n}~\} \ .$$
Each pair ~$\tup{\Sscr_j,\Cscr_j}$~ is of the form:
$$
\tup{~\{F_1@T_1, \ldots, F_{p_j}@T_{p_j}\}, \Cscr_j~} 
$$
\noindent
where $T_1, \ldots, T_{p_j}$ are time variables, $F_1, \ldots, F_{p_j}$ are facts (possibly containing variables) and $\Cscr_j$ is a set of time constraints involving only the variables $T_1, \ldots, T_{p_j}$. 

Given a critical configuration specification, $\CS$, we classify a configuration $\Sscr$ as \emph{critical} w.r.t. $\CS$ if for some $1 \leq i \leq n$, there is a grounding substitution, $\sigma$, such that:
\begin{itemize} 
\item $\Sscr_i \sigma \subseteq \Sscr$; 
\item All constraints in $\Cscr_i \sigma$ are satisfied.
\end{itemize} 
\end{definition}

The substitution application ($\Sscr \sigma$) is defined as usual~\cite{enderton}, \ie, by mapping time variables in $\Sscr$ to natural numbers, nonce names to nonce names (renaming of nonces), and non-time variables to terms.
Notice that nonce renaming is assumed, since the particular nonce name should not matter for classifying a configuration as critical.
Nonce names cannot be specified in advance, since they are freshly generated in a trace, \ie,~during the execution of the process being modelled.

\begin{example}
We can specify usual safety conditions which do not involve time. For example, a drone should never run out of energy. This can be specified by using the following set of critical configuration specification:
\begin{small}
\[
\{~\tup{~\{Dr(Id,X,Y,0)@T\},\emptyset~} \mid Id \in \{d1,d2\}, X \in \{0,\ldots,x_{max}\}, Y \in \{0,\ldots,y_{max}\} ~\} \ .
\]
\end{small}
\end{example}

\begin{example}
The following critical configuration specification specifies a quantitative property involving time:
\begin{small}
\[
\begin{array}{l}
\{~\tup{~\{P(p_1,x_1,y_1)@T_1,Time@T\}, \{ \,T > T_1 + M\, \}~}, \ldots \\ \qquad \qquad \ldots,
\tup{~\{P(p_n,x_n,y_n)@T_n,Time@T\}, \{ \,T > T_n + M\, \}~}~\} \ . 
\end{array}
\]
\end{small}%
Together with the specification in Figure~\ref{fig:rules-complete}, this critical configuration specification specifies that the last pictures of all points of interest ~( \ie, $p_1, \ldots, p_n$ located at $(x_1,y_1),\ldots, (x_n,y_n)$~) should have timestamps no more than $M$ time units old. 
\end{example}

\begin{example}
Let the facts $St(Id,x_b,y_b)@T_1$ and $St(empty,x_b,y_b)@T_1$ denote, respectively, that at time $T_1$ the drone $Id$ entered the base station located at $(x_b, y_b)$ to recharge, and that the station is empty. Moreover, assume that only one drone may be positioned in a station to recharge, which would be specified by adding the following rules specifying the drone landing and take off:
\[
\begin{array}{l}
Time@T,\red{\,Dr(Id,X,Y)@T},\red{\,St(empty, X,Y)@T_1} \lra\\ \quad
Time@T,\blue{\,Dr(Id,X,Y)@(T+1)},\blue{\,St(Id, X,Y)@T}\\[5pt]
Time@T,\red{\,Dr(Id,X,Y)@T},\red{\,St(Id,X,Y)@T_1} \lra \\ \quad
Time@T,\blue{\,Dr(Id,X,Y)@(T+1)},\blue{\,St(empty,X,Y)@T}\\ 
\end{array} 
\]
Then, the critical configuration specification 
$$\{~\tup{\{St(Id,X,Y)@T_1, Time@T\}, \{ \, T > T_1 + M_1\, \} \,} \mid Id \in \{d1,d2\}\,\}$$
specifies that one drone should not remain in a base station for too long (more than $M_1$ time units) preventing other drones to charge.
\end{example}

\begin{example} Fresh values may be  useful in specifying various critical configurations which may involve identification,  history of events 
or communication protocols. For example, drones may communicate between themselves to coordinate their flights. They may also use cryptographic protocols with other agents in the system, \eg, to send pictures of points of interest to be stored on the system data base.
Such applications and requirements are easily formalized using fresh values.\\ For example, drones must be uniquely identified, \ie, should not have the same $Id$:  
\begin{small}
\[
\{~\tup{{~\{Dr(Id,X,Y,E)@T,~\{Dr(Id,X',Y',E')@T'\},\emptyset~} }\} \ .
\]
\end{small}
Also,  in case recharging of batteries is separately managed and billed, even visits to the recharge stations should be uniquely identified for correct billing. Similarly, pictures of points of interest may require identification for documentation. In that case, rules given in Figure~\ref{fig:rules-complete} can easily be modified to include fresh values, \eg, by replacing $P(p_i,X,Y)$ facts with $P(n,p_i,X,Y)$
facts in all rules, and including creation of fresh value $n$ in the rule involving $doClick$ constraint.
\end{example}

\begin{definition}[Compliant Trace]
A trace $\Pscr$ of a timed MSR system is \emph{compliant} w.r.t. a given critical configuration specification $\CS$ if $\Pscr$ does not contain any configuration that is critical w.r.t. $\CS$. 
\end{definition}

Note that if the critical configuration specification is empty,  no configuration is critical, \ie, all traces are compliant.

For simplicity, when the corresponding critical configuration specification is clear from the context, we will elide it and use terminology {\em critical configuration}.
Also, when it is clear from the context, we often elide the timed MSR system and the critical configuration specification with respect to which we consider critical configurations, and simply say that a trace is compliant.

\subsection{Time Sampling}
\label{subsec:time-sampling}

Following \cite{kanovich16formats}, in order to define sensible quantitative verification properties, we need to assume some conditions on when the Tick rule is applicable. Otherwise, any MSR system allows traces containing only instances of $Tick$ rules:
$$
\Sscr_1 \lra_{Tick} \Sscr_2 \lra_{Tick} \Sscr_3 \lra_{Tick} \Sscr_4 \lra_{Tick} \cdots
$$
In such a trace, the system never acts to avoid critical configurations and would easily contain a critical configuration $\Sscr_j$, related to some constraint $T >T' + d$, involving global time $T$ and sufficiently large $j$.

Imposing a \emph{time sampling} is one way to avoid such traces where the time simply ticks. Time sampling is used, for example, in the semantics of verification tools such as Real-Time Maude~\cite{olveczky08tacas}. In particular,  time sampling dictates when the $Tick$ rule must be applied and when it cannot be applied. Such a treatment of time is used for both dense and discrete times in searching and model checking timed systems.

\begin{definition}[Lazy Time Sampling (\ltss)]
\label{def: lazy}
A (possibly infinite) trace $\Pscr$ of a timed MSR system $\Tscr$ uses \emph{lazy time sampling} if for any occurrence of the  $Tick$ rule $\Sscr_i \lra_{Tick} \Sscr_{i+1}$ in $\Pscr$, no instance of any instantaneous rule in $\Tscr$ can be applied to the configuration $\Sscr_i$. 
\end{definition}

In lazy time sampling 
instantaneous rules are given a higher priority than the $Tick$ rule. Under this time sampling, a drone should carry out one of the rules in Figure~\ref{fig:rules-complete} at each time while time can only advance when all drones have carried out their actions for that moment. This does not mean, however, that the drones will satisfy their goal of always having recent pictures of the points of interest as this would depend on the behavior of the system, \ie, the actions carried out by the drones.

In the remainder of this paper, we focus on the lazy time sampling. 
We leave it to future work to investigate whether similar results hold for other time sampling schemes.

\subsection{Verification Problems}
\label{sec:problems}

Four properties are discussed in this section: Realizability and  Survivability from \cite{kanovich16formats} and the new properties of  \newprop\ and \prop. Figure~\ref{fig: properties} illustrates these properties, which we define below.
Since the names of the properties sound similar in English, we also introduce one-letter names for the properties for better readability and differentiation.

The first property we discuss is realizability. It guarantees that the given system can achieve the assigned goal under the given time constraints and design specifications, \eg, that drones can repeatedly collect up-to-date images of the sensitive locations.

Realizability is useful for increasing confidence in a specified system, since a system that is not realizable cannot accomplish the given tasks (specified by a critical specification) and the designer would therefore have to reformulate it. 

However, if a system is shown to be realizable, the trace, $\Pscr$, that proves realizability could also provide insights into the sequence of actions that lead to accomplishment of the specified tasks. This can be used to refine the specification and reduce possible non-determinism.

\vspace{0.5em}
\begin{definition}[Realizability / \realiz property]
\label{def:feasibility}
A timed MSR system $\Tscr$ satisfies \emph{realizability}
with respect to an initial configuration $\Sscr_0$, a critical configuration specification $\CS$ and the \lts 
if there \emph{exists} a compliant infinite time trace from $\Sscr_0$ that uses the \lts
 \footnote{
For simplicity, in the rest of the paper, for properties of systems and configurations, we will not always explicitly state the critical configuration specification, initial configuration, and/or time sampling
with respect to which the property is considered.
For example, when it is clear from the context, we simply say that a system satisfies \realiz property or is \emph{realizable}.\\
Also, when for a property of an MSR $\Tscr$ we only consider traces that use lazy time sampling, we also say that $\Tscr$ \emph{uses the lazy time sampling}.}
 \cite{kanovich16formats}
\end{definition}

The \realiz property of a timed MSR $\Tscr$ w.r.t. $\Sscr_0, \CS$ and \lts can be expressed using the formula:
\begin{equation*}\label{realiz-formula}
    F_\text{\realiz}(\Tscr,\Sscr_0) := \exists \,t \in \textsf{T}^{\Tscr, \Sscr_0}. [ t \in \textsf{T}_{time}^{\Tscr} \cap \textsf{T}_{lts}^{\Tscr} \cap \textsf{T}_c^{\Tscr}],
    \end{equation*}
where ~$\textsf{T}^{\Tscr, \Sscr_0}$ is the set of all  traces of $\Tscr$ starting from $\Sscr_0$,
~$\textsf{T}_{time}^{\Tscr}$ is the set of all infinite time traces of $\Tscr$, ~$\textsf{T}_{lts}^{\Tscr}$ is the set of all traces of $\Tscr$ that use the \lts and ~$\textsf{T}_c^{\Tscr}$ is the set of all traces of $\Tscr$ compliant w.r.t. $\CS$.

Open distributed systems are inherently non-deterministic due to,  \eg, the influence of the environment with winds.
Therefore, it is important to know whether the system can avoid critical configurations despite non-determinism. We call this property \emph{survivability.} 

\begin{definition}[Survivability / \sur property]
\label{def:survivabilty}
A timed MSR $\Tscr$ satisfies \emph{survivability} 
w.r.t. an initial configuration $\Sscr_0$, a critical configuration specification $\CS$ and the \lts 
if it satisfies realizability with respect to $\Sscr_0$, $\CS$, and the \lts
and if \emph{all} infinite time traces from $\Sscr_0$ that use the \lts
are compliant.  
 \cite{kanovich16formats}
\end{definition}

Using the above notation, the \sur property of a timed MSR $\Tscr$ can be expressed with:
\begin{equation*}\label{sur-formula}
F_{\text{\sur}}(\Tscr,\Sscr_0)  :=  F_\text{\realiz}(\Tscr,\Sscr_0) \land \forall \,t \in \textsf{T}^{\Tscr,\Sscr_0}.[ t  \in \textsf{T}_{time}^{\Tscr}  \cap \textsf{T}_{lts}^{\Tscr} \Rightarrow t \in  \textsf{T}_c^{\Tscr}].
    \end{equation*}
    
Although survivability is a desirable property, much more so than realizability, it can sometimes be a rather severe requirement for a system, or even unachievable. Hence, when designing a system, one may want to compromise and consider less demanding properties. 
For example, one may want to avoid configurations that appear as \quotes{dead-ends}, \ie, configurations that necessarily lead to critical configurations. We call such configurations \emph{\pons}.
For example, drones should not fly so far that it is no longer possible to reach a recharging station due to energy consumption.

\begin{figure}[t]
\centering
  \centering
\includegraphics[width={\textwidth}]{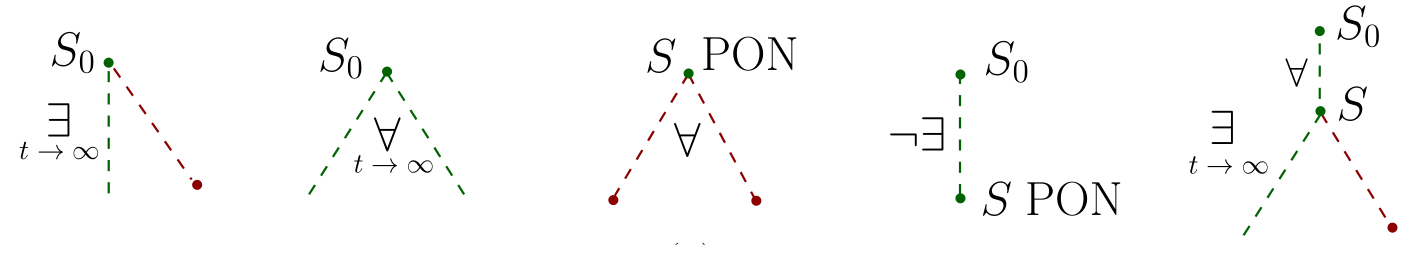}

\vspace{-5pt}
\hspace{0,5em} (a) ~\realiz property \hspace{0,5em} (b) ~\sur property \hspace{0.6 em} (c) \pon \hspace{1em} (d) ~\recov property \hspace{1,25em} (e) ~\rel property \hspace{1em}
\vspace{-2mm}
\label{fig:prop-1}
\caption{Illustration of  properties of (a) realizability, (b) survivability, (d) \prop,~and (e)~\newprop, as well as configurations that are a \pon~(c). 
Green lines represent compliant traces that use lazy time sampling, while red lines represent traces that use lazy time sampling, but are not compliant. Red circles represent critical configurations, while green circles are not critical.
 Quantification marked with $t \to \infty$ denotes quantification over infinite time traces.}
\vspace{-2mm}
\label{fig: properties}
\end{figure}

\begin{definition}[\PON]
\label{def:pon}
Given a timed MSR system $\Tscr$, 
a configuration $\Sscr$ is called a \emph{\pon} with respect to a critical configuration specification $\CS$ and the \lts
if $\Sscr$ is \emph{not} critical with respect to $\CS$, and if 
\emph{all} infinite  traces 
of $\Tscr$  starting with $\Sscr$ and using the \lts 
are not compliant with respect to $\CS$. 
\end{definition}

The set of all configurations that are \pons~of a timed MSR $\Tscr$, $\textsf{C}_{pon}^{\Tscr}$, can be described as 
 ~$ \textsf{C}_{pon}^{\Tscr}  := \{  
\Sscr 
\mid \Sscr \notin  \textsf{C}_{cr}^{\Tscr} \land \forall \,t. [ t \in \textsf{T}^{\Tscr, \Sscr} \cap \textsf{T}_{\infty}^{\Tscr} \cap  \textsf{T}_{lts}^{\Tscr} \ \Rightarrow \ t \notin  \textsf{T}_c^{\Tscr} ]
\},
$
    where 
   ~$\textsf{C}_{cr}^{\Tscr} $ is the set of all critical configurations of $\Tscr$
and ~$\textsf{T}_{\infty}^{\Tscr}$ is the set of all infinite traces of $\Tscr$.
    
There exists no compliant infinite  trace from a \pon ~that uses the \lts 
A \pon~ itself is not critical, but must eventually lead to a critical configuration on every infinite trace that uses the \lts 
Therefore, configurations such as \pons\ are not desirable w.r.t. goal achievement, \ie, \pons\ should be avoided when searching for (infinite) compliant traces.

\begin{remark}
A \pon~ represents the system that still satisfies the required conditions, but it will inevitably fall into a bad state where this is no longer the case.
Therefore, to better distinguish 
between \pons~ and critical configurations,
the condition that a \pon ~is not critical is included in the definition.
\end{remark}

Using the notion of \pons, we introduce new properties of our systems. 

\begin{definition}[\PROP~/ \recov property]
\label{def:prop-nonprop}
 A timed MSR system $\Tscr$, 
satisfies \emph{\prop} with respect to an initial configuration $\Sscr_0$, a critical configuration specification $\CS$ and the \lts
{if it satisfies realizability with respect to $\Sscr_0$, $\CS$ and the \lts
and} 
if no \pon ~
is reachable from $\Sscr_0$ on a compliant trace that uses the \lts
That is, if a configuration $\Sscr$ is reachable from $\Sscr_0$ on a compliant trace that uses the \ltss,
then $\Sscr$ is not a \pon. 
\end{definition}

The \recov property of a timed MSR $\Tscr$ can be expressed with the following formula:
\begin{equation*}\label{recover-formula}
    F_\text{\recov}(\Tscr,\Sscr_0)  := F_\text{\realiz}(\Tscr,\Sscr_0) \land [\forall \,t \in \textsf{T}^{\Tscr, \Sscr_0} \cap \textsf{T}_{c}^{\Tscr} \cap \textsf{T}_{lts}^{\Tscr}. \forall  \Sscr \in t . \Sscr  \notin \textsf{C}_{pon}^{\Tscr}].
  %
    \end{equation*}

Configurations that are \pons\ should be avoided. For example, a drone may enter an area where it may end up with empty batteries due to frequent high winds. Such points  should be avoided.
In fact, with the \recov property we want to ensure that all finite compliant traces from the initial configuration that use the \lts
can be extended to infinite compliant traces that use the \lts

Next, with the \newprop~property, we want to ensure that as long as one follows a compliant trace,  
there is a way to extend the trace to a compliant infinite time trace.
In our drone scenario, a reliable system should be designed so that as long as the drones follow instructions, including rules for flying in  high winds, there is always a way for the drones to avoid critical configurations.

\begin{definition}[\Newprop~/ \rel property]
\label{def:new-prop}
A timed MSR system $\Tscr$ satisfies \emph{\newprop} %
with respect to an initial configuration $\Sscr_0$, a critical configuration specification $\CS$, and the \lts 
{if it  satisfies realizability with respect to $\Sscr_0$,  $\CS$, and the \lts 
and} 
if for \emph{any} configuration $\Sscr$ reachable from $\Sscr_0$ on a \emph{compliant} trace that uses the \ltss, 
there exists a compliant infinite time trace from $\Sscr$ that
uses the \lts 
\end{definition}

The \rel property of a timed MSR $\Tscr$ can be expressed with the following formula:
\begin{equation*}\label{reliab-formula}
\begin{array}{l}
      F_\text{\rel}(\Tscr,\Sscr_0)  := 
    F_\textrm{\realiz}(\Tscr,\Sscr_0) \ \land \\
    \qquad \qquad \quad
    \lbrack \forall  t \in  \textsf{T}^{\Tscr,\Sscr_0} \cap \textsf{T}_{c}^{\Tscr} \cap \textsf{T}_{lts}^{\Tscr}.  
    \forall  \Sscr \in t.
    \exists \,t' \in \textsf{T}^{\Tscr,\Sscr}.  ~t' \in \textsf{T}_{c}^{\Tscr} \cap \textsf{T}_{lts}^{\Tscr} \cup \textsf{T}_{time}^{\Tscr}\rbrack.
    \end{array}
    \end{equation*}

A timed MSR system that satisfies the \rel property represents a system that is always able to avoid \pons. Such a system satisfies the \realiz property, but it may not satisfy the \sur property.
Indeed, the class of systems satisfying the \realiz property is a proper superclass of the class of systems satisfying the \rel property. 
Systems satisfying the \rel property also satisfy the \recov property, while the class of systems satisfying the \recov property ~is a proper superclass of the class of systems satisfying the \sur property. 
We present these results in Section~\ref{sec:relations-properties},  for general MSR systems and  PTSes.


\subsubsection{Time-Bounded Versions of Verification Problems} \ 

\vspace{1em}
Motivated by bounded model checking, we also investigate the time-bounded versions of the above problems.
Instead of infinite traces, in time-bounded versions of the verification problems we consider traces that have exactly a fixed number of occurrences of Tick rules. 
Time bounded version of realizability and survivability were introduced in \cite{kanovich16formats}, while time  bounded version of \newprop~is novel here.

\begin{definition}[$n$-Time Realizability / $n$-\realiz property]
\label{def:n-feasibility} 
A timed MSR system $\Tscr$ satisfies \emph{$n$-\realiz property} with respect to the \ltss, 
a critical configuration specification $\CS$, and an initial configuration $\Sscr_0$ if there \emph{exists} a compliant trace, $\Pscr$, from $\Sscr_0$  that uses the \lts 
such that 
global time advances by exactly $n$ time units in $\Pscr$.
 \cite{kanovich16formats}
\end{definition}

\begin{definition}[$n$-Time Survivability / $n$-\sur property]
\label{def:n-survivabilty}
A timed MSR system $\Tscr$ satisfies \emph{$n$-time survivability property} with respect to the \ltss, 
a critical configuration specification $\CS$ and an initial configuration $\Sscr_0$ if it satisfies {$n$-\realiz property} and if \emph{all} traces with exactly $n$ instances of the $Tick$ rule starting with $\Sscr_0$ and using the \lts 
are compliant. 
\end{definition}

Analogously, we define the $n$-time bounded version of the \newprop ~problem. We consider all compliant traces covering \emph{at most} $n$ time units, and 
extend them to compliant traces over \emph{exactly} $n$ time units.

\begin{definition}[$n$-Time  \Newprop / $n$-\rel property] 
\label{def:n-new-prop}
A timed MSR system $\Tscr$ satisfies \emph{\mbox{$n$-time} \newprop}
with respect to an initial configuration $\Sscr_0$, a critical configuration specification $\CS$, and the \lts 
{if it satisfies {$n$-\realiz property} with respect to $\Sscr_0$,  $\CS$, and the \lts 
and} 
if for any configuration $\Sscr$, reachable from $\Sscr_0$ on a compliant trace $\Pscr$ that uses the \lts 
and has at most $n$ instances of the $Tick$ rule, there exists a trace $\Pscr'$ that uses the \lts 
such that:
\begin{enumerate}
\item[i)] 
 $\Pscr'$ extends $\Pscr$;
\item[ii)] $\Pscr'$ is compliant; 
\item[iii)] $\Pscr'$ has exactly $n$ instances of the $Tick$ rule.
\end{enumerate}
\end{definition}

Since the notion of a \pon\ is defined to be inseparable from infinite traces, it is not appropriate for the time-bounded version of the verification problems. That is, time-bounded version of the \prop\ system problem makes little sense.
Moreover, as we show in Section~\ref{sec:relations-properties}, for PTSes 
problems of \newprop~and 
\prop~coincide.
Hence,  we do not consider the bounded version of \prop\ problem separately.

\section{Relations Among Properties of  Timed MSR} 
\label{sec:relations-properties}

In this section we formally relate all the different properties defined in Section \ref{sec:problems}. 

In order to compare these properties we review the machinery introduced in our previous work~\cite{kanovich.mscs} called $\delta$-representations. 
This machinery is also used  in Section~\ref{sec:complex} to obtain complexity results for the corresponding verification problems.

\subsection{$\delta$-representations}

Some of our results, for a given timed MSR 
$\Tscr$, an initial configuration $\Sscr_0$ and a critical configuration specification $\CS$, will mention the value $\Dmax$ which is an upper-bound on the natural numbers appearing in $\Sscr_0$, $\Tscr$ and $\CS$. The value of $\Dmax$ can be inferred syntactically by simply inspecting the timestamps of $\Sscr_0$, the $D$ values in timestamps of rules (which are of the form \ $T + D$) and constraints in $\Tscr$ and $\CS$ (which are of the form $T_1 > T_2 \pm D$, $T_1 = T_2 \pm D$ and\, $T_1 \geq T_2 \pm D$). For example, the $\Dmax = 1$ for the specification in Figure~\ref{fig:rules-complete}.

For our  results we  assume a bound on the size of facts.
For example, in our specification in Figure~\ref{fig:rules-complete}, we can take the bound
~$ k= |x_{max}| +|y_{max}|+ |e_{max}|+5$.

Notice, however, that we do not always impose an upper bound on the values of timestamps. 
Also, we allow an unbounded number of fresh values to appear in a trace.

\begin{definition}
\label{def: delta-representation}
Let  \, $ \Sscr = \{\,Q_1@t_1, \,Q_2@t_2, \ldots, \,Q_n@t_n \,\} $
be a configuration of a timed MSR $\Tscr$
written in canonical way where the sequence of
timestamps $t_1, \ldots, t_n$ is non-decreasing. 
(For the case of equal timestamps, we sort the facts in alphabetical order, if necessary.) 
The \emph{$\delta$-representation} of $\Sscr$ for a given $\Dmax$ is 
$$
\delta_{\Sscr,\Dmax} = [~Q_1,\,\delta_{Q_1,Q_2},\,Q_2, \ldots, \,Q_{n-1}, \,\delta_{Q_{n-1},Q_n}, \,Q_n~] \ .
$$
Here, for a given natural number $\Dmax$, \ $\delta_{P,Q}$ \ is the \emph{truncated time difference} of two timed facts
~$P@t_1$ and $Q@t_2$~ with \mbox{$t_1\leq t_2$}, defined as follows:
$$
\delta_{P,Q} = 
\left\{\begin{array}{ccl}
t_2 - t_1 & , & \ \textrm{ provided } ~t_2 - t_1 \leq \Dmax\\
\infty & , & \ \textrm{ otherwise } 
\end{array}\right. \ .
$$
\end{definition}

For simplicity, when $\Dmax$ is clear from the context, we sometimes write $ \delta_{\Sscr}$ instead of ~$ \delta_{\Sscr,\Dmax}$.

\vspace{0,5em}
In our previous work~\cite{kanovich12rta,kanovich.mscs}, we showed that a $\delta$-representation is an equivalence class on configurations. Namely, for a given $\Dmax$, we declare $\Sscr_1 $ and $ \Sscr_2$ equivalent, written ~$\Sscr_1 \equiv_{\Dmax} \Sscr_2$, 
if and only if their $\delta$-representations are exactly the same, up to nonce renaming, \ie,~$\Sscr_1 \sigma = \Sscr_2$, where 
$\sigma$ is a bijection on the set of nonce names.

\vspace{0,5em}
This equivalence relation is well-defined with respect to time constrains, \ie, ~configurations that have the same $\delta$-representation satisfy exactly the same set of constraints.
Here, when saying that configurations satisfy the same constraint, we implicitly mean that time variables of the constraint refer to the same facts in both configurations.
Therefore, we can say that a $\delta$-representation satisfies a constraint or does not. Similarly, we say that a $\delta$-representation is critical iff it is the $\delta$-representation of a critical configuration.

Also, the equivalence among configurations is well-defined with respect to application of rules, \ie, ~application of rules on $\delta$-representations is unambiguous. Therefore we can consider traces over $\delta$-representations. 
For details on the concrete procedure of how to apply a rule on a given $\delta$-representation see~\cite[Section 4.3]{kanovich.mscs}.

We naturally extend the notion of a compliant trace and say that a trace over \mbox{$\delta$-representations} is compliant iff it does not contain any critical $\delta$-representation.
Also, we say that a trace over $\delta$-representations uses the \lts 
if $Tick$ rule is applied to a $\delta$-representation in that trace only when no instantaneous rule is applicable.

Moreover, in~\cite[Theorem 4.1]{kanovich.mscs} we have shown that there is a bisimulation between (compliant) traces over configurations and (compliant) traces over their $\delta$-representations in the following sense:
\
~$\Sscr_1 \lra_* \Sscr_2$ \ \ iff \ \ $  \delta_{\Sscr_1} \lra_*   \delta_{\Sscr_2}$ .

When  considering concrete problems and corresponding bisimulations, the bound $\Dmax$ is inferred from numerical values appearing in the problem specification.
This ensures that all configurations in traces are \emph{future bounded}, \ie, do not contain facts $F@t_F$ such that $\delta_{Time,F}= \infty$. This is important for faithful representation of time advances. For more details see~\cite[Section 4.3]{kanovich.mscs}.

For self-containment of the paper, in the proof of the following result from~\cite{kanovich16formats} we present main proof ideas used in~\cite{kanovich.mscs} and, moreover, we additionally address the \lts

\begin{proposition}
\label{thm:delta configurations}
For any timed MSR 
$\Tscr$, a critical configuration specification $\Cscr\Sscr$ and an initial configuration $\Sscr_0$ the equivalence relation between configurations is well-defined with respect to the rules of the system (including time advances), the \lts 
and critical configurations. 
\\
Namely, to any compliant trace starting from the given initial configuration $\Sscr_0$ 
corresponds a compliant trace over $\delta$-representations starting from $\delta_{\Sscr_0}$. 
In particular, a trace over configurations uses the \lts
iff the corresponding trace over $\delta$-representations uses the \lts
~\cite{kanovich16formats}
\end{proposition}

\begin{proof} 
We firstly show that application of rules on $\delta$-representations is independent of the choice of configuration from the same class. Assume 
$\Sscr_1$ and $\Sscr_2$ are equivalent configurations, and assume that $\Sscr_1$ is transformed to $\Sscr'_1$ by means of
a rule~$\alpha$, as shown in the diagram below.
Recall that equivalent configurations satisfy the same set of constraints.
Hence, the rule~$\alpha$ is  applicable to $\Sscr_2$ and will transform 
$\Sscr_2$ into some $\Sscr_2'$:
\[
\begin{array}{cccc}
\Sscr_1 & \to_{\alpha}& \Sscr_1'\\[4pt]
\biginterleave \ 
& & \\[4pt]
\Sscr_2 & \to_{\alpha} \ & \ \Sscr_2'
\end{array}
\]
It remains to show that $\Sscr_1'$ is equivalent to $\Sscr_2'$. 
We consider the two types of rules for $\alpha$, namely, 
time advances and instantaneous rules.

Let the time advance transform
${\Sscr_1}$ into~${\Sscr_1}'$, and $\Sscr_2$ to $\Sscr_2'$.
Since only the timestamp $T$ denoting the global time in $Time@T$ is increased by 1, and the rest of the configuration remains unchanged, 
only truncated time differences involving $Time$ change in the resulting $\delta$-representations. 
Because of the equivalence $S_1 \equiv_{\Dmax} S_2$ , for a fact $P@T_P^1$ in $\Sscr_1$ with $T_P^1\leq T^1$, $Time@T^1$ and $\delta_{P,Time}= t$, we have $ P@T_P^2$ with ${T}_P^2 \leq {T^2}$, $Time@{T^2}$ and $\delta_{P,Time}= t$ in $\Sscr_2$ as well. Therefore, we have
$$\delta_{P,Time}= 
\left\{\begin{array}{ccl}
t+1 & , & \ \textrm{ provided } \ t+1 \leq \Dmax\\
\infty & , & \ \textrm{ otherwise }
\end{array}\right.
$$ 
both in $\Sscr_1'$ and $\Sscr_2'$. On the other hand, for any future fact $Q@T^Q$
with $\delta_{Time,Q}= t$ in $\Sscr_1$ and in $\Sscr_2$, we get $\delta_{Time,Q}= t-1$ in both $\Sscr_1'$ and $\Sscr_2'$.
Therefore, ${\Sscr_1}'$ and $\Sscr_2'$ are equivalent. 
Recall that since all configurations in the trace are future bounded, $t < \infty$, so $t-1$ is well-defined.

The reasoning for the application of instantaneous rules is similar. Each created fact in $\Sscr_1'$ and $\Sscr_2'$ is of the form
$P@(T^1+d)$ and $P@(T^2+d)$ , where $T^1$ and $T^2$ represent
global time in $\Sscr_1$ and $\Sscr_2$, respectively. Therefore
each created fact has the same difference, $d$, to the global time in the corresponding configuration. This implies
that the created facts have the same truncated time
differences to the remaining (unchanged) facts. 
Namely, ~$\delta_{Time,P}=d< \infty$,
~hence for $P@t_P$, $R@t_R$ and $Time@t$ ~with 
~$t \leq t_R \leq t_P$, ~$$\delta_{R,P}= \delta_{Time,P}- \delta_{Time,R}\ . $$
Notice here that $\delta_{Time,R}<\infty$ because all configurations are future bounded, so the above difference is well-defined (finite). 
Similarly, when ~$t \leq t_P \leq t_R$, ~$$\delta_{P,R}= \delta_{Time,R}- \delta_{Time,P}\ .$$
Hence ${\Sscr_1}'$ and
$\Sscr_2'$ are equivalent.
Therefore, application of rules on $\delta$-representations defined through corresponding configurations is well-defined, \ie, the abstraction of configurations to $\delta$-representations w.r.t. application of rules is complete.

The abstraction is also sound. Namely, from a compliant trace over $\delta$-representations, we can extract a concrete compliant trace over configurations.
Although any given \mbox{$\delta$-representation} corresponds to an infinite number of configurations, for a given initial configuration $\Sscr_0$, we have the initial $\delta$-representation
\ $\delta_0= \delta_{\Sscr_0}$.
The existence of a trace over configurations corresponding to the given (possibly infinite) trace over $\delta$-representations is then easily proven by induction. 

Since equivalent configurations satisfy the same set of constraints, 
$\Sscr_1$ is a critical configuration if and only if $\Sscr_2$ is a critical configuration, \ie,~ if and only if $\delta_{\Sscr_1}$ is critical. 
By induction on the length of the (sub)trace, it  follows that, given a timed MSR and a critical configuration specification $\Cscr\Sscr$, 
any (possibly infinite) trace over configurations is compliant if and only if the corresponding trace over $\delta$-representations is compliant.

Notice that, using the \lts 
in a trace $\Pscr$, $Tick$ rule is applied to some $\Sscr_i$ in $\Pscr$ if and only if no instantaneous rule can be applied to 
$\Sscr_i$. Since $\Sscr_i$ and its $\delta$-representation, $\delta_{\Sscr_i}$, satisfy the same set of constraints, it follows that 
$Tick$ rule is applied to $\delta_{\Sscr_i}$ iff~$Tick$ rule is applied to $\Sscr_i$. Hence, a trace over configurations uses the \lts 
iff the corresponding trace over  $\delta$-representations uses the \lts 
\qed
\end{proof}

Following the above result, in the case of balanced timed MSRs, we can work on traces constructed using $\delta$-representations. Moreover, the following lemma~\cite{kanovich16formats} establishes a bound on the number of different $\delta$-representations.

\begin{lemma}\label{lemma:numstates}~\cite{kanovich16formats} 
Let  $\Tscr$ be a timed MSR constructed over a finite alphabet $\Sigma$ with $J$ predicate symbols and $E$ constant and function symbols.
Let  $m$ be the number of facts in the initial configuration $\Sscr_0$, $k$ an upper-bound on the size of the facts,  $\CS$ a critical configuration specification  and   $\Dmax$ an upper-bound on the numerical values of $\Sscr_0, \Tscr$, and $\CS$.
\\
The number of different $\delta$-representations, denoted by $L_\Sigma(m,\uSize,\Dmax)$, is such that
$$
L_\Sigma(m,k,\Dmax) \leq (\Dmax + 2)^{(m-1)} J^m (E + 2 m k)^{m k}.
$$
\end{lemma}

\subsection{Time-Bounded v.s. Unbounded Verification Problems for Timed  MSR }
\label{sec: bounded vs unbounded}

It is obvious, by definition, that the \realiz property implies the \nrealiz property.
We now show that for a sufficiently large $n$, the converse implication also holds, \ie,  the \nrealiz property implies the \realiz property.
The same implications hold for the other properties.

\begin{proposition}[Realizability v.s. $n$-Time Realizability]
\label{th:b-n-bounded-real}
Let $\Tscr$ be a timed MSR that uses  the \ltss, 
$\Sscr_0$  an initial configuration and $\CS$  a critical configuration specification.
   Then,  ~$\Tscr$  satisfies the \realiz property
~ iff ~
$\forall n$, $\Tscr$  satisfies the \nrealiz property.

Moreover,  there exists $M$ such that if  ~$\Tscr$ satisfies the $M$-\realiz property, then ~$\Tscr$ satisfies  the \realiz property.
(In particular, the above claim holds for $M=L_\Sigma(m,\uSize,\Dmax)$.
\end{proposition}
\begin{proof} Per definition, the \realiz property implies the \nrealiz property for any $n$.

We now prove the second statement. The first statement then easily follows.

From Proposition~\ref{thm:delta configurations} it follows that for the above problems we can consider traces constructed over $\delta$-representations. 
As per Lemma \ref{lemma:numstates}, the number of different $\delta$-representa\-ti\-ons
is bounded by ~$l = L_\Sigma(m,\uSize,\Dmax)$, where $m$ is the number of facts in $\Sscr_0$, $k$ is an upper-bound on the size of facts and  $\Dmax$ is an upper-bound on the numeric values of $\Sscr_0, \Tscr$ and $\CS$.

Assume $\Tscr$  satisfies the $M$-\realiz property, where ~$M=L_\Sigma(m,\uSize,\Dmax)$. 
Then,  there is a compliant trace $\Pscr$ from $ \delta_{\Sscr_0}$ that uses the \lts 
and contains  exactly $M$ $Tick$ rules.
Trace $\Pscr$ contains a series of instantaneous  rules separated by $Tick$ rules. That is, $\Pscr$ 
contains  $M+1$ blocks of $\delta$-representations, formed at each of the $M$ instances of $Tick$ rules in $\Pscr$. 
Since there are at most $M$ different $\delta$-representations in $\Tscr$,
 at least one $\delta$-rep\-re\-sen\-ta\-ti\-on $\delta_1$ appears in two blocks.
Therefore, a subtrace between the two appearances of $\delta_1$ contains a $Tick$ rule,
~$
 \delta_{1} \lra \cdots  \lra_{Tick} 
\cdots \lra  \delta_{1} \,
$,
and represents a loop in $\Pscr$.

The above subtrace is compliant, uses the \lts 
and contains a $Tick$ rule. Repeating this loop indefinitely results in a compliant infinite time trace that uses the \lts
The resulting trace shows that $\Tscr$ satisfies the \realiz property.
\qed
\end{proof}

\begin{proposition}[Survivability v.s. $n$-Time Survivability]
\label{th:b-n-bounded-survivability}
Let $\Tscr$ be a timed MSR that uses  the \ltss, 
$\Sscr_0$  an initial configuration 
and $\CS$  a critical configuration specification.
   Then,  ~$\Tscr$  satisfies the \sur property
~ iff ~
$\forall n$, $\Tscr$  satisfies the \nsur property.
\\
Moreover,  there exists $M$ such that if 
 ~$\Tscr$  satisfies the $M$-\sur property, then ~$\Tscr$  satisfies the \sur property.
\end{proposition}
\begin{proof} \ 
Assume that ~$\Tscr$  satisfies the $M$-\sur property, where $M = L_\Sigma(m,\uSize,\Dmax)$. 
Hence, all traces with $M$ ticks are compliant.
Assume $\Tscr$ is does not satisfy the \sur property. Then there is an infinite time trace $\Pscr$ from $\Sscr_0$ that uses the \lts
which  is not compliant, \ie, there is a critical configuration $\Sscr_1$ in $\Pscr$. 
 Because ~$\Tscr$  satisfies the $M$-\sur property, there are more then $M$ ticks in the subtrace  ~$\Sscr_0 \lra \cdots \lra \Sscr_1$~ of $\Pscr$.

Since there are more than  $M$ rules (and $\delta$-representations) in the above subtrace $\Pscr'$, some $\delta$-representation appears at least twice in $\Pscr'$, \ie,
 there is a loop in $\Pscr'$. By removing all loops in $\Pscr'$  we obtain a  trace $\Pscr''$ from ~$\Sscr_0$ to $\Sscr_1$ that uses the \lts 
 and contains at most $M$ rules. Consequently, there are at most $M$ ticks  in $\Pscr''$. 

The trace $\Pscr''$ is not a compliant  since it contains $\Sscr_1$. This is in contradiction with the $M$-\sur property of $\Tscr$.

Per definitions, the \sur property implies the \nsur property for any $n$.
The first statement is then a simple consequence of the second statement and definitions.
\qed
\end{proof}

\begin{proposition}[\Newprop~v.s. $n$-Time \Newprop]
\label{th:b-n-bounded-recoverability}
Let $\Tscr$ be a   timed MSR that uses  the \ltss, 
$\Sscr_0$  an initial configuration 
and $\CS$  a critical configuration specification.
   Then,  ~$\Tscr$  satisfies the \rel property
~ iff ~
$\forall n$, $\Tscr$  satisfies the \nrel property.
\\
Moreover,  there exists $M$ such that if 
  ~$\Tscr$  satisfies the $M$-\rel property, then ~$\Tscr$  satisfies the \rel property.
\end{proposition}
\begin{proof} 
As in Proposition \ref{th:b-n-bounded-real}, we consider traces constructed over $\delta$-representations. We show the second statement. The first statement follows from the second statement and definitions.

Let ~$M = L_\Sigma(m,\uSize,\Dmax)$.  
 Assume $\Tscr$  satisfies the $M$-\rel property. 
{Then $\Tscr$  satisfies the $M$-\rel property and, by the proof of Proposition~\ref{sec: bounded vs unbounded}, it satisfies the \rel property. 
Namely, a compliant trace $P$ with $M$ $Tick$ rules that uses the \lts
is its own compliant extension, showing the  $M$-\rel property.
 }
 
  We still need to show that an arbitrary compliant trace from $ \delta_{\Sscr_0}$ that uses the \ltss, 
  $\Pscr$,  can be extended to a compliant infinite time trace that uses the\lts 
  Let $m$ be the number of $Tick$ rules in $\Pscr$.

 
 Since $\Tscr$  satisfies the $M$-\rel property, $\Pscr$ can be extended to a compliant trace $\Pscr'$ that uses the \lts 
 and contains exactly $M$ $Tick$ rules.
  As in the proof of  Proposition \ref{th:b-n-bounded-real}, we can conclude that there is a \mbox{$\delta$-representation} $ \delta'$ in $\Pscr'$ that appears at least twice with a $Tick$ rule between the two appearances of  $\delta'$ in $\Pscr'$.
Repeating the subtrace of $\Pscr'$ between  two appearances of  $\delta'$ indefinitely, creates an infinite time compliant trace from $\delta_{\Sscr_0}$ that uses the \lts 
and shows the \rel property of $\Tscr$.
\qed
\end{proof}

\subsection{Relations Among Different Properties of Timed MSR and PTS}
\label{sec:relations-properties-infinite}

In this Section we formally relate different properties defined over infinite traces. In general, we can distinguish all these properties for timed MSR, but only some for PTSes, as stated below.

\begin{proposition}
\label{thm:deadlock-recover}
Let $\Tscr$ be a timed MSR system that uses the \ltss, 
$\Sscr_0$  an initial configuration  and $\CS$  a critical configuration specification.

If ~$\Tscr$ satisfies the \rel property, then $\Tscr$ satisfies the \recov property.

If ~$\Tscr$ satisfies the \recov property, then $\Tscr$ does not necessarily satisfy the \rel property.


\end{proposition}
\begin{proof}
Let $\Tscr$ be a  timed MSR system that satisfies  the \rel property.
\\
 Assume $\Tscr$ does not satisfy the \recov property.
Then, since $\Tscr$ satisfies  the \realiz property, there is a  compliant trace from $\Sscr_0$ to some \pon~$\Sscr_P$  that uses the \lts 
Since $\Tscr $ satisfies  the \rel property, there is a  compliant infinite time trace from $\Sscr_P$ that uses the \lts
As $\Sscr_P$ is a  \pon, this contradicts the notion of  \pon.
\\[3pt]
We  give an example of a timed MSR system, $\Tscr$, that satisfies  the \recov property, but does not that satisfy  the \rel property. 
\\[3pt]
 Let ~$\Sscr_0'= \{ Time@0, C@1\}$, ~$\CS'= \emptyset$,
and  let  $\Tscr'$ contain only the following instantaneous rules:
\begin{subequations}
\begin{align}
Time@T, \,\red{C@T'} \mid T' \leq T \ \lra \ \ Time@T,\, \blue{\,D@T}
\label{eq:ex03-0}
\\Time@T, \,\red{C@T'} \mid T' \leq T \ \lra \ \ Time@T,\, \blue{\,A@T}
\label{eq:ex03-1}
\\
Time@T, \, \red{A@T'} \  \ \lra \ \ Time@T,\, \blue{\,B@T}
\label{eq:ex03-2}
\\Time@T, \,\red{B@T'}\  \ \lra \ \ Time@T,\, \blue{\,A@T}
\label{eq:ex03-3}
\end{align}
\end{subequations}
The system $\Tscr'$ satisfies  the \realiz property since there is a compliant infinite time trace from $\Sscr_0'$ that uses the \ltss: 
\begin{equation}
\begin{small}
\begin{array}{ll}
\label{ex3: infinite c trace}
Time@0, \, {C@1} \ \lra_{Tick} \
Time@1, \, {C@1}\ \lra_{(\ref{eq:ex03-0})} \
Time@1, \, {D@1}\ \lra_{Tick} \
\\ \quad
Time@2, \, {D@1}\ \lra_{Tick} \
Time@3, \, {D@1}\ \lra_{Tick} \
Time@4, \, {D@2}\ \lra_{Tick} \
\dots 
\end{array}
\end{small}
\end{equation}
There is only one other infinite trace form $\Sscr_0$ that uses the \ltss: 
\begin{equation}
\begin{small}
\begin{array}{ll}
\label{ex3: infinite trace}
Time@0, \, {C@1} \ \lra_{Tick} \
Time@1, \, {C@1}\ \lra_{(\ref{eq:ex03-1})} \
Time@1, \, {A@1}\ \lra_{(\ref{eq:ex03-2})} \
\\ \quad
Time@1, \, {B@1}\ \lra_{(\ref{eq:ex03-3})} \
Time@1, \, {A@1}\ \lra_{(\ref{eq:ex03-2})} \
Time@1, \, {B@2}\ \lra_{(\ref{eq:ex03-3})} \
\dots 
\end{array}
\end{small}
\end{equation}
Its subtrace obtained from $\Sscr_0$ by applying the $Tick$ rule followed by the rule $({\ref{eq:ex03-1}})$
reaches the configuration $Time@1, A@1$. 
This subtrace is compliant but it cannot be extended to a compliant infinite time trace that uses the \lts 
Hence, $\Tscr'$ does not satisfy  the \rel property.
\\[2pt]
However, $\Tscr'$ trivially satisfies  the \recov property ~since there are no critical configurations and, hence, no \pons.
\qed
\end{proof}

\vspace{1em}
Notice that the properties of timed MSR defined in Section~\ref{sec:problems} involve infinite time traces that use the \lts 
Recall that for any given PTS $\Tscr$ and any configuration $\Sscr$, 
there exists an infinite time trace of $\Tscr$ that starts with $\Sscr$ and uses the \lts 

Although  \recov and \rel are different properties  of timed MSR systems in general, it turns out that  for the class of  PTSes these properties coincide.


\vspace{0,5em}
\begin{proposition}
\label{thm:deadlock-recover-PTS}
Let $\Tscr$ be a PTS that uses the \ltss, 
$\Sscr_0$  an initial configuration,  and $\CS$  a critical configuration specification.

System ~$\Tscr$  satisfies the \rel property ~ iff  ~$\Tscr$ satisfies  the \recov property.

\end{proposition}
\begin{proof}
Since a PTS is a timed MSR system, it follows from Proposition \ref{thm:deadlock-recover} that a PTS, which satisfies  the \rel property, also satisfies  the \recov property.
\\[3pt]
Assume that a PTS $\Tscr$ does not satisfy  the \rel property. If $\Sscr_0$ is critical, then $\Tscr$ does not satisfy  the \realiz property and consequently, does not satisfy the \recov property.
If $\Sscr_0$ is not critical, there is a compliant trace  from $\Sscr_0$ to some configuration 
$\Sscr_1$ that uses the \lts 
which cannot be extended to a compliant infinite time  trace that uses the \lts 

Then, $\Sscr_1$ is a \pon.
Namely, if $P$ is an infinite trace from $\Sscr_1$ that uses the \ltss,
by Proposition~\ref{prop:progressing}, $P$ is an infinite time trace that uses the \lts 
Then, $P$ is not compliant.  
%
Since the \pon~$\Sscr_1$ is reachable from $\Sscr_0$ on a  compliant trace using the \ltss, 
$\Tscr$ does not satisfy  the \recov property.
\qed
\end{proof}

We  show that the remaining properties are different even for PTSes. Furthermore, we  show relations among the 
properties for PTSes and  for timed MSR systems in general.
We first show that \rel~and \sur are different properties of PTSes, and, consequently,  different properties of timed MSR systems.

\begin{proposition}
\label{thm:deadlock-survive}
Let $\Tscr$ be a PTS that uses the \ltss, 
$\Sscr_0$  an initial configuration  and $\CS$  a critical configuration specification.

If ~$\Tscr$ satisfies  the \sur property, then $\Tscr$ satisfies  the \rel property.

If $\Tscr$ satisfies  the \rel property, it may not satisfy  the \sur property.


\end{proposition}
\begin{proof}
Assume that $\Tscr$ satisfies  the \sur property, but does not satisfy  the \rel property. 
 Then, {since $\Tscr$ satisfies  the \realiz property,}  there exists a compliant trace, $\Pscr$, from $\Sscr_0$ to some configuration $\Sscr_1$  that cannot be extended to a compliant infinite time trace  that uses the \lts 
Let $\Pscr'$ be an infinite time trace which is an extension of $\Pscr$  that uses the \lts 
Such a trace $\Pscr'$ exists due to Proposition~\ref{prop:progressing}, but it is not compliant.
\\
Since $\Tscr$ satisfies  the \sur property, all infinite time traces from $\Sscr_0$ that use the \lts 
are compliant, including $\Pscr'$. 
 Contradiction.
\\[5pt]
The following example of a PTS  satisfies  the \rel property, but  does not satisfy  the \sur property. 

\noindent
Let ~$\Sscr_0= \{ Time@0, A@0, B@0\}$, ~$\CS= \{~\tup{~\{B@T, D@T'\},\emptyset~} \}$ and let PTS $\Tscr$ contain only the following instantaneous rules:
\begin{subequations}
\begin{small}
\begin{align}
Time@T, \,\red{A@T'}, \, B@T''  \mid   \{ T' \leq T, T'' \leq T\}  &\lra  \, Time@T,\, B@T'', \blue{\,C@(T + 1)}
\label{eq:ex-1}
\\
Time@T, \,\red{A@T'}, \, B@T''  \mid  \{ T' \leq T\}  &\lra \, Time@T,\, B@T'', \blue{\,D@(T + 1)}
\label{eq:ex-2}
\\
Time@T, \,\red{B@T'}, \, \red{C@T''}  \mid  \{T' \leq T, T''\leq T \}  &\lra  \, Time@T,\, \blue{A@T}, \blue{\,B@(T + 1)}
\label{eq:ex-3}
\end{align}
\end{small}
\end{subequations}
The following trace from $\Sscr_0$ uses the \lts 
and is not compliant:

\vspace{2pt} 
\(
Time@0, A@0, B@0\ \lra_{(\ref{eq:ex-2})} \ Time@0, B@0, D@1 \ .
\)\\[2pt]
Hence, $\Tscr$ does not satisfy the \sur property.

To show that $\Tscr$ satisfies the \rel property, we first show that $\Tscr$  satisfies the \realiz property. The following trace from $\Sscr_0$ is a compliant infinite time trace that uses the \ltss: 
$$
{\small
\begin{array}{l}
    Time@0, A@0, B@0 \ \lra_{(\ref{eq:ex-1})} \ \ Time@0, B@0, C@1 \lra_{Tick}    \\[2pt] 
    \ \ \lra_{Tick}     \ \  Time@1, B@0, C@1 \  \lra_{(\ref{eq:ex-3})} \ \ Time@1, A@1, B@2 \ \lra_{Tick}     
    \\[2pt] \ \ \ \  \lra_{Tick} \ \ 
       Time@2, A@1, B@2 \ \lra_{(\ref{eq:ex-1})} \ \ Time@2, B@2, C@3 \ \lra_{Tick}  \dots
\end{array}
}
$$
Next, assume $\Pscr$ is a compliant trace from $\Sscr_0$ to some $\Sscr_1$ that uses the \lts 
Then $\Pscr$ does not contain rule (\ref{eq:ex-2}), which always results in a critical configuration.
Hence, only rules (\ref{eq:ex-1}), (\ref{eq:ex-3}) and $Tick$  are used in $\Pscr$, so $\Sscr_1$ is either 
\ $
\{ \, Time@t, \, {A@t'}, \, B@t'' \, \}$  \text{ or } \ $\{ \, Time@t, \, B@t', \,{C@t''}\, \}.
$\
Using only the rules (\ref{eq:ex-1}), (\ref{eq:ex-3}) and $Tick$, the trace $\Pscr$ can be extended to a compliant infinite time trace that use the \lts 
Hence,  $\Tscr$ satisfies the \rel property.
\qed
\end{proof}

However, the above does not hold for general MSR systems, 
\ie,  MSR systems that satisfy the \sur property do not necessarily satisfy the \rel property.


\vspace{10pt}
\noindent
\begin{proposition} \label{thm:deadlock-survive-msr}
Let $\Tscr$ be a timed MSR that uses the \ltss, 
$\Sscr_0$  an initial configuration  and $\CS$  a critical configuration specification.

If ~$\Tscr$ satisfies the \sur property, it may not satisfy the \rel property.

If $\Tscr$ satisfies the \rel property, it may not satisfy the \sur property.


\end{proposition}
\begin{proof}
Let $\Tscr'$, $\Sscr_0'$ and $\CS'$ be as specified in the proof of Proposition \ref{thm:deadlock-recover}.
Recall that $\Tscr'$ does not satisfy the \rel property.

The system $\Tscr'$ satisfies the \sur property. Namely, there are only two infinite traces from $\Sscr_0'$ that use the \ltss, 
traces (\ref{ex3: infinite c trace}) and (\ref{ex3: infinite trace}) specified in  the proof of Proposition \ref{thm:deadlock-recover}.
However, trace (\ref{ex3: infinite trace}) is not an infinite time trace, so there is only one infinite time trace from $\Sscr_0'$ that uses the \ltss, 
trace (\ref{ex3: infinite c trace}).
Therefore, since trace (\ref{ex3: infinite c trace}) is compliant, $\Tscr'$ satisfies the \sur property.

By Proposition~\ref{thm:deadlock-survive} there is a PTS, and therefore an MSR, that satisfies the \rel property but does not satisfy the \sur property.
\qed
\end{proof}


Next, we show how the \recov property relates to the \realiz property.

\begin{proposition}
\label{thm:realize-recover}
Let $\Tscr$ be a timed MSR that uses the \ltss, 
$\Sscr_0$  an initial configuration,  and $\CS$  a critical configuration specification.

If ~$\Tscr$ satisfies the \recov property, then $\Tscr$ satisfies the \realiz property.

A system $\Tscr$ that satisfies the \realiz property may not satisfy the \recov property.


\end{proposition}

\begin{proof}
Assume $\Tscr$ satisfies the \recov property. 
Then, $\Tscr$ satisfies the \realiz property  by definition.

We prove the other statement by providing an example of a PTS that satisfies the \realiz property, but does not satisfy the \recov property.
Let ~$\Sscr_0''= \{ Time@0, A@0\}$, ~$\CS''= \{~\tup{\{D@T\},\emptyset~} \}$ and let PTS $\Tscr''$ contain only the following instantaneous rules:
\begin{subequations}
\begin{small}
\begin{align}
Time@T, \,\red{A@T'}\ \mid \ \, \{ T' \leq T\} \ \lra \ \ Time@T,\, \blue{\,B@(T + 1)}
\label{eq:ex1-3}
\\
Time@T, \,\red{A@T'} \ \mid \ \, \{ T' \leq T\} \ \lra \ \ Time@T,\, \blue{\,C@(T + 1)}
\label{eq:ex1-1}
\\
Time@T, \, \red{B@T'} \ \mid \ \, \{ T'\leq T \} \ \lra \ \ Time@T,\, \blue{\,A@(T + 1)}
\label{eq:ex1-4}
\\
Time@T, \,\red{C@T'} \ \mid \ \, \{ T' \leq T\} \ \lra \ \ Time@T,\, \blue{\,D@(T + 1)}
\label{eq:ex1-2}
\end{align}
\end{small}
\end{subequations}
The following trace, which uses the \ltss, 
shows the \realiz property of $\Tscr''$:

\vspace{2pt} 
\(
 \begin{small}
\begin{array}{l}
Time@0,\,{A@0} \, \lra_{(\ref{eq:ex1-3})} \, Time@0, {\,B@1} \ \lra_{Tick} \, Time@1, {\,B@1} 
\ \lra_{(\ref{eq:ex1-4})} \\
\ \ Time@1, {\,A@2} \  \lra_{Tick} \, Time@2, {\,A@2} \ \lra_{(\ref{eq:ex1-3})} \, Time@2, {\,B@3} \ \lra_{Tick} \ \dots
\end{array} 
 \end{small}
\)\\[2pt]
 The configuration $\widetilde{\Sscr}=\{Time@0, C@1\}$ is reachable from $\Sscr_0''$ by a compliant trace that uses the \ltss: 
~$
Time@0,\,\red{{A@0}} \ \lra_{(\ref{eq:ex1-1})} \ Time@0, \blue{{\,C@1} }  .
$
~$\widetilde{\Sscr}$ is a \pon~as rule 
(\ref{eq:ex1-2}) is the only instantaneous rule that can be applied after a $Tick$, so all infinite traces from $\widetilde{\Sscr}$ that use the \lts 
contain the critical configuration~ \mbox{$\{Time@1, D@2\}$}. 

Since $\widetilde{\Sscr}$ is a \pon, 
 $\Tscr''$ does not satisfy the \recov property.
\qed
\end{proof}

\begin{remark}
Notice that requiring the \realiz property in the definition of the \rel property~(Definition~\ref{def:new-prop}) is redundant in the sense that the \realiz property would follow from the \rel property anyway. 

Namely, the non-critical initial configuration $\Sscr_0$ is trivially reachable from $\Sscr_0$ on a compliant trace that uses the \lts 
Then, the \rel property of $\Tscr$ would imply the existence of a compliant infinite time trace from $\Sscr_0$ that uses the \lts 
%
 Hence, $\Tscr$ would satisfy the \realiz property.

Similarly, 
the \realiz property would follow from the \recov property, even if it was not required by definition (Definition~\ref{def:prop-nonprop}), provided that the initial configuration was not critical.
 Namely, the set of compliant traces from a critical configuration is empty, so in the case of a critical initial configuration, the system would not satisfy the \realiz property but
 the \recov property would hold because of universal quantification over the empty set.

 Assuming that $\Sscr_0$ is not critical, $\Sscr_0$ is trivially reachable from $\Sscr_0$ on a compliant trace that uses the \lts
 Then, the \recov property of $\Tscr$ implies that $\Sscr_0$ is not a \pon. 
Then, as per definition of a \pon, there is a compliant infinite trace, $\Pscr$, from $\Sscr_0$ that uses the \lts 
As per Proposition~\ref{prop:progressing}, $\Pscr$ is  a compliant infinite time trace from $\Sscr_0$ that uses the \lts 
 Hence, $\Tscr$ satisfies the \realiz property.

Therefore, it was not necessary to include the condition of the \realiz property in Definitions~\ref{def:prop-nonprop} and \ref{def:new-prop}. We do so to keep these notions intuitive in the sense that these properties are more restrictive that the \realiz property.
\end{remark}

Using transitivity of the subset relation, we can infer relations among all our properties for both PTSes and timed MSR systems in general.
We summarize our results in the following corollaries.

\begin{corollary}
\label{thm:properties}
Let ~$^{~}_{reali}Z_{ability}^{~MSR}$ ,  $^{~}_{re}L_{iability}^{~MSR}$ , $^{~}_{reco}V_{erability}^{~MSR}$~ and  ~$S_{urvivability}^{~MSR}$~ be  the classes of timed MSR systems  satisfying the \realiz, \rel, \recov and \sur properties, respectively, w.r.t. the \lts 
Then, the following 
relations hold:
$$
S_{urvivability}^{~MSR}
\ {\neq}
^{~}_{re}L_{iability}^{~MSR}
\ \subset \ 
^{~}_{reco}V_{erability}^{~MSR}
\ \subset \ 
^{~}_{reali}Z_{ability}^{~MSR}
$$
\end{corollary}
\begin{proof}
The statement follows directly from  the Propositions {\ref{thm:deadlock-survive-msr}}, \ref{thm:deadlock-recover},  
 and \ref{thm:realize-recover}. 
\qed
\end{proof}

\begin{corollary}
\label{thm:properties-PTS}
 Let ~$^{~}_{reali}Z_{ability}^{~PTS}$ ,  $^{~}_{re}L_{iability}^{~PTS}$ , $^{~}_{reco}V_{erability}^{~PTS}$~ and  ~$S_{urvivability}^{~PTS}$~
be the classes of PTSes satisfying the \realiz, \rel, \recov and \sur properties, respectively, w.r.t. the \lts 
Then the following proper subset relations hold:
$$
S_{urvivability}^{~PTS}
\ \subset \
^{~}_{re}L_{iability}^{~PTS}
\ = \ 
^{~}_{reco}V_{erability}^{~PTS}
\ \subset \ 
^{~}_{reali}Z_{ability}^{~PTS}
$$
\end{corollary}
\begin{proof}
The statement follows directly from  the Propositions  \ref{thm:deadlock-survive} and \ref{thm:deadlock-recover-PTS}, and the proof of proposition \ref{thm:realize-recover}. %
\qed
\end{proof}

\begin{corollary}
\label{thm:properties-n}
 Let ~$^{~}_{reali}nZ_{ability}^{~PTS}$ ,  $^{~}_{re}nL_{iability}^{~PTS}$~ and  ~$nS_{urvivability}^{~PTS}$~
 be  the classes of PTSes  satisfying  the \nrealiz, \nrel and \nsur properties, respectively, w.r.t. the \lts 
Then, the following proper subset relations hold:
$$
n\Sscr_{urvivability}^{~PTS} 
 \ \subset \ ^{~}_{re}nL_{iability}^{~PTS}\ \subset \ ^{~}_{reali}nZ_{ability}^{~PTS} \ .
$$
\end{corollary}

\begin{proof}
Let a PTS $\Tscr$ satisfy the \nsur property. 
We  check that $\Tscr$ satisfies the \nrel property.
Let  $\Sscr$ be a configuration that is reachable from $\Sscr_0$ on a compliant trace $\Pscr$ that uses the \lts 
and has at most $n$ instances of the $Tick$ rule. Since $\Tscr$ is a PTS, only a bounded number of instantaneous rules can be applied before a $Tick$ rule appears in a trace that uses the \lts 
(Proposition~\ref{prop:bounded-length}). Hence, the trace $\Pscr$ can be extended to a compliant trace $\Pscr'$  that contains exactly $n$ instances of the $Tick$ rule and uses the \lts 
Since $\Tscr$   satisfies the \nsur  property,  $\Pscr'$ is compliant. Consequently,   $\Tscr$ satisfies the \nrel property.

Now, let  $\Tscr$ satisfy the \nrel property. Then, the trivial trace of length 1 from $\Sscr_0$ (containing only $\Sscr_0$) can be extended to a  compliant trace $\Pscr'$  that contains exactly $n$ instances of the $Tick$ rule and uses the \lts 
Hence, $\Tscr$ satisfies the \nrealiz property.

To show that the inclusions are proper, we give examples of PTSes that satisfy one, but not the other property.
The PTS given in the proof of Proposition~\ref{thm:realize-recover} is an example of a system that satisfies the \nrealiz property, $\forall n >0$, which does not satisfy even the $1$-\sur property.
Similarly, the PTS given in the proof of Proposition~\ref{thm:deadlock-survive} satisfies the \nrel property, $\forall n >0$, but it does not even satisfy the $1$-\sur property.
\qed
\end{proof}

\section{Complexity Results for PTSes } 
\label{sec:complex}
In this section we investigate the complexity of the verification problems  defined in Section~\ref{sec:problems} for progressing timed systems.
Recall that the \rel and \recov properties for PTS coincide.


\subsection{PSPACE-Completeness of Verification Problems for PTSes}

We again point out that in our previous work~\cite{kanovich13ic,kanovich.mscs,kanovich17jcs,kanovich21jcs} we only dealt with {finite traces}. The additional challenge in addressing the complexity of the \realiz and \sur problems
in \cite{kanovich16formats} and the new verification problems  \rel and  \recov 
in this paper,  is to deal with infinite traces.

\vspace{0,5em}
Assume throughout this section the following: 
\begin{itemize}
\item $\Sigma$ -- A finite alphabet with $J$ predicate symbols and $E$ constant and function symbols;
\item $\Tscr$ -- A progressing timed MSR 
constructed over $\Sigma$;
\item $\Sscr_0$ -- An initial configuration;
\item $m$ -- The number of facts in the initial configuration $\Sscr_0$;
\item $\CS$ -- A critical configuration specification constructed over $\Sigma$;
\item $k$ -- An upper-bound on the size of facts;
\item $\Dmax$ -- An upper-bound on the numeric values of $\Sscr_0, \Tscr$ and $\CS$.
\end{itemize}

\vspace{1em}
For the \realiz problem and the \sur problem, PSPACE-completeness for PTSes was proved in \cite{kanovich16formats,kanovich16arxivformats}. 
PSPACE-hardness of the \rel problem for PTSes 
can be shown by adequately adapting our previous work~\cite{kanovich11jar,kanovich16arxivformats}.

\vspace{0,5em}
\begin{proposition}[
\rel problem for PTSes is PSPACE hard]
\label{th: realizability-hard}
\ \\
The \rel problem
for PTSes 
that use the \lts 
is PSPACE-hard.
\end{proposition}
\begin{proof}
The \realiz problem is an instance of the problem of checking whether a configuration  is not a \pon. 
Recall that a system satisfies the \realiz property if there exists a compliant infinite time trace $\Pscr$ from the initial configuration in which global time tends to infinity.
Since $\Tscr$ is progressing, we obtain the condition on time (time tends to infinity) from Proposition~\ref{prop:progressing}.
Indeed, a system satisfies the \realiz property  if and only if the initial configuration is not a \pon.
Since PSPACE and co-PSPACE are the same complexity class 
and the \realiz property is PSPACE-hard, the problem of determining whether a configuration is a \pon ~ is PSPACE-hard.

Since the  \rel property 
problem comprises checking whether a configuration is a \pon, 
it is PSPACE-hard.
\qed
\end{proof}

\vspace{0,5em}
Infinite traces over configurations of PTSes are infinite time traces (Proposition~\ref{prop:progressing}). The same holds for traces over 
$\delta$-representations of PTSes.
For our verification problems we, therefore, need to construct an infinite compliant trace. The following lemma establishes a criterion.

\begin{lemma}
\label{lem:lengthPSPACE}
For a PTS $\Tscr$ assume $\Sigma, \Sscr_0, m,\CS,k,\Dmax$ as described above. If there is a compliant trace (constructed using $\delta$-representations) starting with the $\delta$-representation of $\Sscr_0$ with length $L_\Sigma(m,k,\Dmax)$, then there is an \emph{infinite} compliant trace starting with the $\delta$-representation of $\Sscr_0$. 
~\cite{kanovich16formats}
\end{lemma}
\begin{proof}
Since there are only $L_\Sigma(m,k,\Dmax)$ different \mbox{$\delta$-representations}, a trace of length greater than $L_\Sigma(m,k,\Dmax) $ necessarily contains the same \mbox{$\delta$-representations} twice, that is, there is a loop in the trace. By repeating the $\delta$-representations appearing in the loop, we can construct an infinite trace which is necessarily compliant. 
\qed
\end{proof}

For our complexity results we use some auxiliary functions with configurations or \mbox{$\delta$-representations} as their arguments, as suitable.
For our next results involving the four properties,
assume that for any given timed MSR $\Tscr$, an initial configuration $\Sscr_0$ and a critical configuration specification $\CS$
we have two functions, $\next$ and $\critical$, which check, respectively, whether a rule in $\Tscr$ is applicable to a given $\delta$-representation and whether a $\delta$-representation is critical with respect to $\CS$. Moreover, for a given timed MSR $\Tscr$, let \mustTick\ be a function implementing the \lts 
a $\delta$-representation of
system $\Tscr$, and returns 1 when the $Tick$ must be applied and 0 when it must not be applied according to the \lts 
We assume that $\next$, $\critical$, and \mustTick\ run in Turing space bounded by a polynomial in $m,k,\log_2(\Dmax)$. Notice that for our examples this is the case, as such functions can be constructed because the system is balanced and facts are of bounded size.

Because of Lemma~\ref{lem:lengthPSPACE}, we can show that the \realiz problem is in PSPACE by searching for compliant traces of length $L_\Sigma(m,k,\Dmax)$ (stored in binary). To do so, we rely on the fact that PSPACE and NPSPACE are the same complexity class~\cite{savitch}.

\begin{proposition}[Realizability for PTSes is in PSPACE]
\label{th:PSPACE-feasibility}
\ \\
For a PTS $\Tscr$ assume $\Sigma, \Sscr_0, m,\CS,k,\Dmax$,  $\next, \critical$, and \mustTick\ as described above. 

There is an algorithm that, given an initial configuration $\Sscr_0$, decides whether $\Tscr$ satisfies the \realiz property with respect to the \ltss, 
$\CS$, and $\Sscr_0$ and the algorithm runs in a space bounded by a polynomial in $m,k$ and $log_2(\Dmax)$. 

The polynomial is in fact ~$\log_2(L_\Sigma(m,k,\Dmax))$.
~\cite{kanovich16formats}
\end{proposition}

\vspace{0,5em}
We now consider the \sur problem. Recall that in order to prove that $\Tscr$ satisfies the \sur property  with respect to the \ltss, 
a critical configuration specification $\CS$ and an initial configuration $\Sscr_0$, we must show that $\Tscr$ satisfies the \realiz property w.r.t. the \ltss, 
$\CS$ and $\Sscr_0$,
and that all infinite traces $\Pscr$ that use the \lts 
and start with $\Sscr_0$ are compliant with respect to $\CS$, and that the global time in $\Pscr$ tends to infinity (Definition~\ref{def:survivabilty}).

Checking that a system satisfies the \realiz property is PSPACE-complete as we have just shown. Moreover, the property that time tends to infinity follows from Proposition~\ref{prop:progressing} for progressing timed MSR. It remains to show that all infinite traces using the \lts 
are compliant, which reduces to checking that \emph{no critical configuration is reachable} from the initial configuration $\Sscr_0$ by a trace that uses the \lts 
This property can be decided in PSPACE by relying on the fact that PSPACE, NPSPACE and co-PSPACE are all the same complexity class~\cite{savitch}. Therefore, the \sur problem is also in PSPACE. 

\vspace{1em}
\begin{proposition}[Survivability for PTSes is in PSPACE]
\label{th:PSPACE-survivability}
\ \\
For a PTS $\Tscr$ assume $\Sigma, \Sscr_0, m,\CS,k,\Dmax$,  $\next, \critical$, and \mustTick\ as described above. 

There is an algorithm that decides whether $\Tscr$ satisfies the the \sur property with respect to the \ltss, 
$\CS$ and $\Sscr_0$ which runs in space bounded by a polynomial in $m,k$ and $log_2(\Dmax)$. 
~\cite{kanovich16formats}
\end{proposition}

We now investigate the complexity of the \rel problem for PTSes, that is for the problem of deciding whether a given PTS satisfies the \rel property.
This problem coincides with the \recov problem for PTSes. 
Recall that a configuration $\Sscr$ is a \pon~ 
iff ~there is no compliant infinite time trace
from $\Sscr$ that uses the \lts 

\vspace{1em}
\begin{proposition}[\Newprop ~for PTSes is in PSPACE] 
\label{th:PSPACE-newprop}
\ \\
For a PTS $\Tscr$ assume $\Sigma, \Sscr_0, m,\CS,k,\Dmax$,  $\next, \critical$, and \mustTick\ as described above. 

There is an algorithm that, given an initial configuration $\Sscr_0$, decides whether $\Tscr$ satisfies the \rel property with respect to the \ltss,
$\CS$ and $\Sscr_0$ and the algorithm runs in space bounded by a polynomial in $m,k$ and $log_2(\Dmax)$. 
\end{proposition}
\begin{proof}
We first propose an algorithm that, for a fixed system $\Tscr$ and a fixed critical configuration specification $\CS$,
checks whether some $\delta$-representation corresponds to a configuration that  is a \pon ~w.r.t. $\Tscr$ and $\CS$.

Following PSPACE-completeness of the \realiz problem~\cite{kanovich16formats,kanovich16arxivformats}, let $REAL$ denote the \realiz problem PSPACE algorithm over $\delta$-representations.
(For details, see the proof of \cite[Theorem 1]{kanovich16arxivformats}.) 
When given $\delta$-representation 
$\Wscr$, as input, the algorithm $REAL(\Wscr)$ returns ACCEPT if and only if there is an infinite time trace of $\Tscr$ that starts with $\Wscr$, uses the \ltss, 
and is compliant w.r.t. $\CS$,
and it
runs in polynomial space w.r.t. the given parameters, $m,k$ and $log_2(\Dmax)$.
Since PSPACE 
and co-PSPACE are the same complexity class~\cite{savitch}, we 
switch the ACCEPT and FAIL and obtain a deterministic algorithm $NOTREAL$ that runs in polynomial space w.r.t. $m,k$ and $log_2(\Dmax)$. The algorithm $NOTREAL(\Wscr)$ accepts if and only if 
there is no compliant infinite time trace from $\Wscr$ that uses the \lts 

Then, using $NOTREAL$ we construct the algorithm $PON$ that checks whether the given $\delta$-representation $\Wscr$ corresponds to a \pon. Let $PON$ be the following algorithm, which takes a $\delta$-representation $\Wscr$ as input:
\begin{quote}
\begin{enumerate}
\item If \ $\critical(W) = 1$, \ie, if \ $W$  represents a critical configuration, then return FAIL, otherwise continue;
\item If \ $NOTREAL(W) = 1$, 
\ie, if \ $W$  represents a \pon, 
then return ACCEPT, otherwise return FAIL.
\end{enumerate}
\end{quote}

When given $\delta$-representation 
$\Wscr$, as input, the algorithm $PON(\Wscr)$ accepts if and only if $\Wscr$ is a $\delta$-representation of a \pon ~w.r.t. $\Tscr$ and $\CS$. 
Since $\critical $ and $NOTREAL$ run in the polynomial space w.r.t. $m,k$ and $log_2(\Dmax)$,
$PON$ is a deterministic algorithm that also runs in such a polynomial space.

\vspace{2pt}
Next, we check that for any configuration $\Sscr$ reachable from $\Sscr_0$ using the \ltss, 
there is a compliant infinite time trace from $\Sscr$, \ie, 
that $\Sscr$ is not a \pon.
The following algorithm accepts when no \pon ~is reachable from $\Sscr_0$ in $\Tscr$ on a compliant trace that uses the \ltss, 
and fails otherwise.
It begins with $i=0$ and $W_0$ set to be the $\delta$-representation of $\Sscr_0$, and iterates the following sequence of operations:
\begin{quote}
\begin{enumerate}
\item If \ $W_i$  represents a critical configuration, \ie, if $\critical(W_i) = 1$, then return FAIL, otherwise continue;
\item If \ $W_i$ represents a \pon, \ie, if $PON(W_i) = 1$, then return FAIL, otherwise continue;
\item If \ $i > L_\Sigma(m,k,\Dmax) $, then ACCEPT;
else continue;
\item If \mustTick $(W_i)=1$, then replace $W_i$ by $W_{i+1}$ obtained from $W_i$ by applying the $Tick$ rule; Otherwise  non-deterministically guess an instantaneous rule, $r$, from $\Tscr$ applicable to $W_i$,
\ie, such a rule $r$ that $\next(r,W_i) = 1$. If so,
replace $W_i$ with the $\delta$-representation $W_{i+1}$
resulting from applying the rule $r$ to the $\delta$-representation
$W_i$. 
Otherwise, continue;
\item Set \ $i = i + 1$.
\end{enumerate}
\end{quote}

Since $PON$, $\next, \critical$ and \mustTick\ 
run in Turing space bounded by a polynomial in $m$, $k$ and $\log_2(\Dmax)$, 
it follows that the above algorithm runs in deterministic polynomial space.
\qed
\end{proof}


\begin{theorem}
Let $\Sigma$ be a finite alphabet, $\Tscr$ a PTS, 
$\Sscr_0$ an initial configuration, $m$ the number of facts in $\Sscr_0$, $\CS$ a critical configuration specification, $k$ an upper-bound on the size of facts, and $\Dmax$ an upper-bound on the numeric values in $\Sscr_0, \Tscr$ and $\CS$.
Let the functions $\next, \critical$, and \mustTick\ 
run in Turing space bounded by a polynomial in $m,k,\log_2(\Dmax)$ and return 1, respectively, when a rule in $\Tscr$ is applicable to a given $\delta$-representation, when a $\delta$-representation is critical with respect to $\CS$, and when a $Tick$ rule should be applied to the given $\delta$-representation using the \lts 

The \rel problem for PTSes that use the \lts 
is PSPACE-complete when assuming a bound on the size of facts.
\end{theorem}
\begin{proof}
The result follows directly from Propositions 
\ref{th: realizability-hard}
and \ref{th:PSPACE-newprop}.
\qed
\end{proof}

\begin{remark}
The PSPACE-completeness of the \realiz and \sur problems
obtained in \cite{kanovich16formats} 
follows directly from Propositions ~\ref{th:PSPACE-feasibility}
and ~\ref{th:PSPACE-survivability}, and the lower bound proof from \cite[Appendix C]{kanovich16arxivformats}. 
\end{remark}

\subsection{Complexity Results for Time-Bounded Problems }

We now consider the $n$-time versions of the problems defined in Section \ref{sec:problems}.

\vspace{3pt}
The following lemma establishes an upper-bound on the length of traces with exactly $n$ instances of $Tick$ rules for PTSes. This bound is immediately obtained from Proposition~\ref{prop:bounded-length}.
\begin{lemma}
\label{lem:polysize}
Let $n$ be fixed.
 Let $\Tscr$ be a PTS,
$\Sscr_0$ an initial configuration and $m$ the number of facts in $\Sscr_0$. 
For all traces $\Pscr$ of $\Tscr$ that start with $\Sscr_0$ and contain exactly $n$ instances of the Tick rule, the length of $\Pscr$ is bounded by ~$(n+1)*m+n$.
~\cite{kanovich16formats}
\end{lemma}
\begin{proof}
Assume that there are exactly $n$ time ticks in $\Pscr$. 
As per Proposition \ref{prop:bounded-length} there are at most $m-1$ instantaneous rules between any $Tick$ and the next $Tick$ rule in $\Pscr$. Consequently, in total there are at most $(n+1)\cdot m + n$ rules in $\Pscr$.
\qed
\end{proof}

\begin{remark}
For our bounded versions of the \realiz, \rel and \sur problems
we use auxiliary functions with configurations as their arguments.
For simplicity, we still use the same notation, $\next, \critical$, and $\mustTick$ , as types of arguments are clear from the context.
We will assume that $\next$, $\critical$, and \mustTick\ run in time bounded by a polynomial in $m$ and $k$,
 and return 1, respectively, when a rule in $\Tscr$ is applicable to a given configuration, when a configuration is critical with respect to $\CS$, and when a $Tick$ rule should be applied to the given configuration using the \lts 
Notice that for our examples one can construct such functions.
\end{remark}

For NP-hardness, we encode the NP-hard problem 3-SAT as an the \nrealiz problem similar to our previous work~\cite{kanovich13esorics,kanovich10fcsPriMod}. 

\begin{proposition}
\label{th: n-realizability-hard}
The \nrealiz problem for PTSes
w.r.t. the \lts 
is NP-hard.
~\cite{kanovich16formats,kanovich16arxivformats}
\end{proposition}%
\begin{proof}
Assume we are given a formula \ $ \widetilde{C}_n = (l_{11} \lor l_{12} \lor l_{13}) \land \cdots \land (l_{n1} \lor l_{n2} \lor l_{n3})$
in 3-CNF, \ie,~in a conjunctive normal form with exactly 3 literals in each clause. Recall that each literal, $l_{ij}$, is either an atomic formula, $v_k$ or 
its negation, $\neg v_k$.

Let $p$ be the number of variables and $n$ the number of clauses in $F$. 

We construct an initial configuration $\Sscr_0$ and a progressing timed MSR $\Tscr$ that uses the \lts 
to check whether $F$ is satisfiable or not.

We set  ~$\CS  = \emptyset$~ and the initial configuration as $$\Sscr_0 = \{ \ Time@0,\,V_1@0,\ldots,\,V_p@0,\,I_{\widetilde{C}_n})@0 \ \}\ .$$
 Here, and henceforth, the fact $I_C$ represents a conjunctive
 normal form~$C$.

With $V_1$,\dots,$V_n$, we generate,
 non-deterministically, the truth
 values to be assigned to  $x_1$,\dots,$x_n$:  
\begin{equation}
\label{eq-vars}
\begin{array}{@{}l}
 Time@T,\ \red{V_i@T_i} \, \mid\, \{~T_i \leq T~\} \ \lra\
 Time@T,\ \blue{A_{i,1}@(T+1)}
\\ 
 Time@T, \ \red{V_i@T_i} \, \mid\, \{~T_i \leq T~\} \ \lra\
 Time@T,\ \blue{A_{i,0}@(T+1)}
\end{array}
\end{equation}%
 where $A_{i,1}$ stands for: {\em ``1~is assigned to~$x_i$, ''}
 and $A_{i,0}$ means: {\em ``0~is assigned to~$x_i$. ''}
Intuitively, these rules construct an interpretation for the variables in $F$.

 To cut our presentation, we introduce the following abbreviation
\begin{equation}%
    x_i^{\varepsilon} :=
\left\{\begin{array}{rl}{}%
        x_i,          &   \quad \mbox{for\ \ $\varepsilon = 1$}.
\\{}%
        \neg x_i,    &   \quad \mbox{for\ \ $\varepsilon = 0$}.
\end{array}\right.
                                         \label{eq-epsilon}
\end{equation}%
For each disjunct \mbox{$(x_{k_1}^{\varepsilon_1}\lor x_{k_2}^{\varepsilon_2}
 \lor x_{k_3}^{\varepsilon_3})$},
 let  $A_{k_1,\varepsilon_1}$,  $A_{k_2,\varepsilon_2}$,
  $A_{k_3,\varepsilon_3}$ be the facts that makes the disjunct ``true''.
Notice that there are only three such facts for each disjunct.

\noindent
 Now we ``compute'' the value of the CNF~$\widetilde{C}_n$
 starting from the left-most disjuncts by repeatedly erasing
 the correct disjuncts, replacing\
\mbox{$I_{(x_{k_1}^{\varepsilon_1}
 \lor x_{k_2}^{\varepsilon_2}
 \lor x_{k_3}^{\varepsilon_3})\land C}$}
 with the shorter~$I_{C}$,
 as follows
\begin{equation}
\label{eq-compute-true}
\begin{array}{@{}l}
 Time@T, \,\red{A_{k_1,\varepsilon_1}@T_1},\
 \red{I_{(x_{k_1}^{\varepsilon_1}
 \lor x_{k_2}^{\varepsilon_2}
 \lor x_{k_3}^{\varepsilon_3})\land C}@T_2 }
\, \mid \,\{\,T_1 \leq T,\ T_2 \leq T \,\} 
\\       \hspace*{\fill} \lra\
 Time@T, \,\blue{A_{k_1,\varepsilon_1}@T_1},\ \blue{I_{C}@(T+1)}
\\
 Time@T, \,\red{A_{k_2,\varepsilon_2}@T_1},\
 \red{I_{(x_{k_1}^{\varepsilon_1}
 \lor x_{k_2}^{\varepsilon_2}
 \lor x_{k_3}^{\varepsilon_3})\land C}@T_2} 
\, \mid \,\{\,T_1 \leq T,\ T_2 \leq T \,\} 
\\       \hspace*{\fill} \lra\
 Time@T, \,\blue{A_{k_2,\varepsilon_2}@T_1},\ \blue{I_{C}@(T+1)}
\\
 Time@T, \,\red{A_{k_3,\varepsilon_3}@T_1},\
 \red{I_{(x_{k_1}^{\varepsilon_1}
 \lor x_{k_2}^{\varepsilon_2}
 \lor x_{k_3}^{\varepsilon_3})\land C}@T_2} 
\, \mid \,\{\,T_1 \leq T,\ T_2 \leq T \,\} 
\\       \hspace*{\fill} \lra\
 Time@T, \,\blue{A_{k_3,\varepsilon_3}@T_1},\ \blue{I_{C}@(T+1)}
\end{array}
\end{equation}%

By inspection, the constructed $\Tscr$ is a progressing timed MSR. 

The CNF~$\widetilde{C}_n$ is satisfiable iff there is a 
 trace that in $2n$~time steps leads from~$S_0$
 into a configuration, which include \mbox{$I_{\emptyset}@T$}. Namely, 
 we spend $n$ time steps, see~\eqref{eq-vars},
 to generate the truth values,
 and the next $n$ time steps, see~\eqref{eq-compute-true},
 to get the correct values of all $n$ disjuncts.


Hence, the $2n$-\realiz property problem is NP-hard.
\qed
\end{proof}

\vspace{0.5em}
For the \nsur  problem we obtain a new lower bound complexity result.

\begin{proposition}
\label{th: n-survivability-hard}
The \nsur problem for PTSes w.r.t. the \lts 
is both NP-hard and coNP-hard.
\end{proposition}%
\begin{proof}
Recall the 3-SAT encoding from the proof of Proposition \ref{th: n-realizability-hard}, showing that the \nrealiz problem is NP-hard.
Since in the encoding the critical configuration specification is empty,  $\CS = \emptyset$, all traces are necessarily compliant.
 Therefore,
 the \mbox{$n$}-\sur problem is NP-hard as well.

\paragraph{\bf coNP-hardness}
In order to prove that the \nsur problem is coNP-hard,
 we will modify our original system in the following way.

Let  all configurations that include \mbox{$I_{\emptyset}@T$}
  are declared critical, \ie, let \ $\CS'= \{ I_{\emptyset}@T \} $.
 To circumvent the \realiz property issues, we add the
 ``dummy''  rules, \mbox{$j=0,1,2,\dots 2n$}:
\begin{equation}
\label{eq-H}
\begin{array}{@{}l}
 Time@T,\ H_j@T_j \, \mid\, \{~T_j \leq T~\} \ \lra\
 Time@T,\ H_{j+1}@(T+1)
\end{array}
\end{equation}%
 Also, ~$H_0$ is added to
 the initial configuration:
\begin{equation}%
 S_0'=\{\ Time@0,\ V_1@0,\ ...,\ V_n@0,\ I_{\widetilde{C}_n}@0,\ H@0 \}
                                              \label{eq-init-H}
\end{equation}%

 The \mbox{$2n$}-\realiz problem is trivial because of~\eqref{eq-H}.

 If CNF~$\widetilde{C}_n$ is satisfiable then there is a 
 trace~$\tau$ that in $2n$~time steps leads from~$S_0$
 into a configuration~$S$, which include \mbox{$I_{\emptyset}@T$}. Naemly, we spend $n$ time steps, see~\eqref{eq-vars},
 to generate the truth values,
 and the next $n$ time steps, see~\eqref{eq-compute-true},
 to get the correct values of all $n$ disjuncts.

 Since $S$ is critical, the trace~$\tau$ is not compliant,
 and, hence, our updated system does not satisfy \mbox{$2n$}-\sur property.

 If CNF~$\widetilde{C}_n$ is not satisfiable then there is no
 trace that leads from~$S_0$
 into a configuration~$S$, which include \mbox{$I_{\emptyset}@T$}.
 Hence, our system satisfies \mbox{$2n$}-\sur property,
 and the \mbox{$n$}-\sur problem is coNP-hard.

\noindent
 Bringing the things together we conclude that 
 the \nsur problem is both NP-hard and coNP-hard.
\qed
\end{proof}

\vspace{0.5em}
The \nrealiz problem is in NP as stated in the following Proposition. 

\begin{proposition}[\nrealiz problem is in NP]
\label{thm:np-feasible}
\ \\
For a PTS $\Tscr$ assume $\Sigma, \Sscr_0, m,\CS,k,\Dmax$,  $\next, \critical$, and \mustTick\ as described above. 
\\
The \nrealiz problem for $\Tscr$ w.r.t. the \ltss, 
$\CS$, and $\Sscr_0$ is NP-complete.
%
~\cite{kanovich16formats,kanovich16arxivformats}
\end{proposition}
%
%
%
%
%

\vspace{0,5em}
We now provide the upper bound complexity  result for the \nsur problem for progressing timed MSRs.

Recall that for the \nsur property, we need to show that:
\begin{enumerate}
\item $\Tscr$ satisfies \nrealiz property with respect to $\CS$;
\item All traces using the \lts 
with exactly $n$ ticks are compliant with respect to $\CS$.
\end{enumerate}

\begin{proposition}[\nsur problem is in $\Delta_2^p$ ]
\label{th: n-time-survivability}
\ \\
For a PTS $\Tscr$ assume $\Sigma, \Sscr_0, m,\CS,k,\Dmax$,  $\next, \critical$, and \mustTick\ as described above. 

The \nsur problem
for $\Tscr$
w.r.t.  the \ltss, 
$\CS$ and $\Sscr_0$ is in the class 
$\Delta_2^p$ of the polynomial hierarchy ($P^{NP}$) with input $\Sscr_0$.
~\cite{kanovich16formats,kanovich16arxivformats}
\end{proposition}

Finally, we provide the complexity 
upper bound for the \nrel problem for PTSes.

\begin{proposition}[\nrel problem is in $\Pi_2^p$ ]
\label{th: n-time-newprop}\ \\
For a PTS $\Tscr$ assume $\Sigma, \Sscr_0$, $m$, $\CS$, $k$, $\Dmax$,  $\next, \critical$, and \mustTick\ as described above. 

The \nrel  problem 
for $\Tscr$ w.r.t. the \ltss, 
$\CS$ and $\Sscr_0$ is in  
 \mbox{$\Pi_2^p$}, the second class 
of the polynomial hierarchy 
with input $\Sscr_0$.
\end{proposition}
\begin{proof}
 We formalize the bounded reliability problem
 with the following predicate \ \mbox{$L_n(S_0,S_1,\dots,S_n)$},
 which means that
\begin{quote}
 `` The configurations $S_0$, $S_1$, $S_2$, dots, $S_n$, form a trace
 $$ S_0 \rightarrow S_1 \rightarrow S_2 \rightarrow \dots
 \rightarrow S_n $$%
 in which all $S_0$, $S_1$, $S_2$, dots, $S_n$ are not critical.''
\end{quote}

\noindent
 With the initial~$S_0$,
 then the system is n-L reliable iff (realizability)
$$ \exists S'_1\ \exists S'_2\ \exists\ \dots\ \exists S'_n
 L_n(S_0,S'_1,\dots,S'_n) $$%
 and (reliability) for any~$k$ such that \mbox{$k=0,1,\dots,n$},
 the following holds: 
$$ \forall S_1\  \forall S_2\  \forall\dots\ \forall S_k\ 
( L_k(S_0,S_1,\dots,S_k) \longrightarrow
\exists S_{k+1}\ L_{k+1}(S_0,S_1,\dots,S_k,S_{k+1}))
 $$%
 Hence we can rewrite the above statement in the form of $\Pi_2^p$,
 the second class in the polynomial hierarchy, 
$$ \forall x_1\  \forall x_2\  \forall\dots\ \forall x_m\
   \exists y_1\  \exists y_2\  \exists\dots\  \exists y_m\
   Q(x_1, x_2, \dots, x_m, y_1, y_2, \dots, y_m)
$$%
 where\  \mbox{$Q(x_1, x_2, \dots, x_m, y_1, y_2, \dots, y_m)$}\
 is recognizable in polynomial time.

\qed
\end{proof}

The obtained complexity result for the time-bounded verification problem of \rel problem is stated in the next theorem.

\begin{theorem}
\label{thm:np-comp}
Assume $\Sigma$ a finite alphabet, $\Tscr$ a PTS, 
$\Sscr_0$ an initial configuration, $m$ the number of facts in $\Sscr_0$, $\CS$ a critical configuration specification and $k$ an upper-bound on the size of facts.
Let functions $\next, \critical$, and \mustTick\ 
run in Turing time bounded by a polynomial in $m$ and $k$ 
 and return 1, respectively, when a rule in $\Tscr$ is applicable to a given configuration, when a configuration is critical with respect to $\CS$, and when a $Tick$ rule should be applied to the given configuration.
 
 The  \nrealiz problem for $\Tscr$ w.r.t. the \ltss, 
 $\CS$, and $\Sscr_0$ is NP-complete with input $\Sscr_0$.

The \nsur  problem for $\Tscr$ w.r.t. the \ltss, 
$\CS$, and $\Sscr_0$ is both NP-hard and coNP-hard with input $\Sscr_0$. Furthermore, the \nsur problem for $\Tscr$ w.r.t. the \ltss, 
$\CS$, and $\Sscr_0$ is in the
class $\Delta_2^p$ of the polynomial hierarchy ($P^{NP}$) with input $\Sscr_0$.

The  \nrel problem for $\Tscr$ w.r.t. the \ltss, 
$\CS$ and $\Sscr_0$ is in the class 
$\Pi_2^p$ of the polynomial hierarchy 
with input the $\Sscr_0$.

%
%
%
\end{theorem}

\section{Related and Future Work } 
\label{sec:related}

In this paper, we study a subclass of timed MSR systems called progressing timed systems introduced in \cite{kanovich16formats}, which is defined by imposing syntactic restrictions on MSR rules. 

We discuss two verification problems, namely realizability and survivability, introduced in \cite{kanovich16formats}, and also consider two new properties, \newprop~and  \prop, defined over infinite traces. We show that these problems are PSPACE-complete for progressing timed systems, and when we additionally impose a bound on time, the realizability becomes NP-complete,  the survivability is in $\Delta_2^p$, the \newprop~is in $\Pi_2^p$ class of the polynomial hierarchy, and  the survivability is both NP-hard and coNP-hard. The lower bound for the $n$-time reliability is left for future work.

These problems involve quantitative temporal properties of timed systems and explicit time constraints,  and to the best of our knowledge have not been studied in the rewriting literature.
%
%
We review some of the formalisms for specifying quantitative temporal properties of timed systems 
such as timed automata, temporal logic, and rewriting. 
%

Others have proposed languages for specifying properties that allow explicit time constraints. We review some of these formalisms, such as timed automata, temporal logic, and rewriting. 

Our progressing condition is related to the {finite-variability assumption} used in the temporal logic and timed automata literature~\cite{faella08entcs,laroussinie03tcs,lutz05time,alur91rex,alur04sfm}, requiring that  in any bounded interval of time, there can be only finitely many observable events or state changes. Similarly, progressing systems have the property that only a finite number of instantaneous rules can be applied in any bounded interval of time (Proposition~\ref{prop:bounded-length}). Such a property seems to be necessary for the decidability of many temporal verification problems.

The work~\cite{alpern87dc,clarkson10jcs} classifies (sets of) traces as safety, liveness, or properties that can be reduced to subproblems of safety and liveness. Following this terminology, properties that relate to our verification problems over infinite traces contain both elements of safety and elements of liveness. 
 We do not see how this can be expressed precisely in terms of~\cite{alpern87dc,clarkson10jcs}. 
 We leave this investigation to future work.

As discussed in detail in the Related Work section of our previous work~\cite{kanovich.mscs}, there are some important differences between our timed MSR model and timed automata~\cite{alur91rex,alur04sfm}, both in terms of expressive power and decidability proofs. For example, a description of a timed MSR system uses first order formulas with variables, whereas timed automata can only refer to transitions on ground states. That is, timed MSR is essentially a first-order language, while timed automata are propositional. Replacing a first order description of timed MSR by all its instantiations, would lead to an exponential explosion. Furthermore, in contrast with the timed automata paradigm, in timed MSR we can naturally manipulate facts both in the past, in the future, as well as in the present.

The temporal logic literature has proposed many languages for the specification and verification of timed systems. 
Many temporal logics contain quantitative temporal operators, \eg,~\cite{lutz05time,laroussinie03tcs}, including time-constrained temporal operators. 
 Metric Temporal Logic (MTL)~\cite{koymans90} involves (bounded or unbounded) timing constraints on temporal operators similar to our time constraints. The growing literature on MTL explores the expressivity and  decidability of fragments of such temporal logics
~\cite{Ouaknine06TACAS}.
However, the temporal logic literature does not discuss notions similar to \eg, realizability or survivability.  
In addition to that, an important feature of our model is that the specifications are executable. 
As we have shown through experiments in~\cite{kanovich16formats}, it is feasible to analyze fairly large progressing systems using the rewriting logic tool Maude~\cite{clavel-etal-07maudebook}.

Real-Time Maude~\cite{olveczky08tacas} is a tool for simulation and analysis of real-time systems. Rewrite rules are partitioned into
instantaneous rules and rules that advance time, with
instantaneous rules taking priority. 
Our lazy time sampling is inspired by such management of time in traces.
Time advance rules in Real-Time Maude
may place a bound on the amount of time to advance, but do
not determine a specific amount, thus, allowing the system to be continuously observed.
 Time sampling strategies are
used to implement search and model-checking analysis.
{\"{O}}lveczky and Messeguer~\cite{olveczky07entcs} investigate conditions under
which the maximal time sampling strategy used in Real-Time
Maude is complete. One of the required conditions is
tick-stabilizing, which is similar to progressing and the
finite variability assumption in that one assumes a bound
on the number of actions that can be applied in a finite time.

Cardenas~\etal~\cite{cardenas08icdcs} discuss possible verification problems of cyber-physical systems in the presence of malicious intruders. They discuss surviving attacks, such as denial of service attacks on control mechanisms of devices. We believe that our progressing timed systems can be used to define meaningful intruder models and formalize the corresponding survivability notions. This may lead to automated analysis of such systems similar to the successful use of the Dolev-Yao intruder model~\cite{DY83} for protocol security verification. Given the results of this paper, the decidability of any security problem would most likely involve a progressing timed intruder model.
We intend to investigate the security aspects of this work in the future. For example, the introduction of timed intruder models~\cite{kanovich13ic}, and resource-bounded intruder models~\cite{kanovich21jcs} may enable verification of whether intruders can cause PTSes to reach hazardous situations, \eg, harm to people or crashes.

We believe that some of our properties, in particular survivability, can be interpreted using game theory. 
We find that our model has some features that are better suited for applications relating to TSDSes, in particular explicit time, quantitative time conditions and nonces.
It would be interesting to investigate connections and differences between our rewriting approach to these problems and the game theory approach.

Finally, we have already done some preliminary research into ways to extend this work to dense time domains. We expect our results to hold for dense time domains as well, given our previous work~\cite{kanovich17jcs,kanovich19csf,kanovich21jcs}. There, instead of the $Tick$ rule (Eq.~\ref{eq:tick}), we assume a $Tick$ rule of the form ~$Time@T \lra Time@(T + \varepsilon)$, where $\varepsilon$ can be instantiated by any positive real number.  
The assumption of dense time is a challenge that considerably increases the machinery needed to prove our results, but we are confident of finding ways to combine the results of~\cite{kanovich17jcs} with those presented in this paper.  
Similarly, for our future work, we intend to investigate extensions of our models with probabilities.

\vspace{5mm}

\emph{Acknowledgments:}
 Ban Kirigin is supported in part by the Croatian Science Foundation under the project UIP-05-2017-9219. 
 The work of Max Kanovich was partially supported by EPSRC Programme Grant EP/R006865/1: “Interface Reasoning for Interacting Systems (IRIS).”
 Nigam is partially supported by NRL grant N0017317-1-G002, and CNPq grant 303909/2018-8. 
 Scedrov was partially supported by the U. S. Office of Naval Research under award numbers N00014-20-1-2635 and N00014-18-1-2618. 
 Talcott was partially supported by the U. S. Office of Naval Research under award numbers  N00014-15-1-2202 and N00014-20-1-2644, and NRL grant N0017317-1-G002.





\end{document}